\documentclass[letterpaper, 10 pt, draftcls, onecolumn]{ieeeconf}  

\IEEEoverridecommandlockouts                              
\usepackage{graphicx} 
\usepackage{amsmath,bm,times} 
\usepackage{amssymb}
\usepackage{tabularx}
\usepackage{algpseudocode}
\usepackage{tikz}

\usetikzlibrary{shapes,arrows,backgrounds,fit,positioning}

\newtheorem{mydef}{\textbf{Definition}} [section]
\newtheorem{thm}[mydef] {\textbf{Theorem}}
\newtheorem{lemma}[mydef]{\textbf{Lemma}}
\newtheorem{prop}[mydef]{\textbf{Proposition}}

\newtheorem{remark}[mydef]{\textbf{Remark}}

\allowdisplaybreaks[1]

\title{\LARGE \bf
A Class of LTI Distributed Observers for LTI Plants: Necessary and Sufficient Conditions for Stabilizability
}

\author{Shinkyu Park and Nuno C. Martins
\thanks{}
\thanks{Shinkyu Park and Nuno C. Martins are with the Department of Electrical and Computer Engineering, University of Maryland College Park, College Park, MD 20742-4450, USA.
        {\tt\small \{skpark, nmartins\}@umd.edu}}%
}

\begin{document}

\maketitle

\begin{abstract}
Consider that an autonomous linear time-invariant (LTI) plant is given and that a network of LTI observers assesses its output vector. The dissemination of information within the network is dictated by a pre-specified directed graph in which each vertex represents an observer. Each observer computes its own state estimate using only the portion of the output vector accessible to it and the state estimates of other observers that are transmitted to it by its neighbors, according to the graph. This paper proposes an update rule that is a natural generalization of consensus, and for which we determine necessary and sufficient conditions for the existence of parameters for the update rule that lead to asymptotic omniscience of the state of the plant at all observers. The conditions reduce to certain detectability requirements that imply that if omniscience is not possible under the proposed scheme then it is not viable under any other scheme that is subject to the same communication graph, including nonlinear and time-varying ones.
\end{abstract}

\section{Introduction} \label{sec_intro}
Consider the following linear time-invariant (LTI) plant with state $x(k)$ and output vector $y(k)$\footnote{In order to simplify our notation, without loss of generality, we omit noise terms in the state-space equation \eqref{eq_LTI_plant}. See \textbf{(i)} of Subsection \ref{subsection_properties} for more discussion.}:
\begin{equation} \label{eq_LTI_plant}
	\begin{split}
		x(k+1) = Ax(k)\\
		y(k) = Hx(k) \\
	\end{split}
\end{equation}
\begin{equation*}
	\begin{split}
		\text{where } &y(k) = \left ( y_1^T(k), \cdots, y_m^T(k) \right )^T \text{ with } y_i(k) = H_i x(k), \\
		&x(k) \in \mathbb{R}^{n}, y_i(k) \in \mathbb{R}^{r_i}
	\end{split}
\end{equation*}

Let $\mathcal{G} = \left( \mathbb{V}, \mathbb{E} \right)$ be a directed graph. Each vertex in $\mathbb{V}$ represents an observer and each edge in $\mathbb{E} \subseteq \mathbb{V} \times \mathbb{V}$ determines the viability and direction of information exchange between two observers. Each observer computes a state estimate based on a portion of the output of the plant and state estimates of the other observers connected to it via an edge of $\mathcal{G}$. We refer to $\mathcal{G}$ and the collection of all observers as a \textit{distributed observer} (see Fig. \ref{figure_framework}).

Let $\hat{x}_i(k)$ be a state estimate by observer $i$ at time $k$. A distributed observer is said to achieve \textit{omniscience asymptotically} if it holds that $\lim_{k \rightarrow \infty} ||\hat{x}_i(k) - x(k)|| = 0$ for all $i \in \mathbb{V}$, i.e., the state estimate at every observer converges to the state of the plant.

\tikzstyle{plant} = [draw, text width=5em, text centered, rounded corners, minimum height=5em, thick]
\tikzstyle{observer} = [draw, text width=5em, text centered, rounded corners, minimum height=2em, thick]
\tikzstyle{observer_e} = [text width=5em, text centered]
\tikzstyle{vertex} = [circle, draw, inner sep=0, minimum width=.5em, thick, fill=black]

\def\blockdist{2.4}
\def\edgedist{2.5}
\def\ybias{-10}

\begin{figure} [t]
	\centering
	\begin{tikzpicture} \label{figure_framework}
		\node (plant) [plant] {LTI Plant};
		\path ([yshift=-\ybias] plant.north east)+(\blockdist,0) node (ob1) [observer] {Observer 1};
		\path (plant.east)+(\blockdist,0.1) node (obmid) [observer_e] {\LARGE $\vdots$};
		\path ([yshift=\ybias] plant.south east)+(\blockdist,0) node (obm) [observer] {Observer m};
		
		\path (ob1.east)+(1.5,0.3) node [vertex] (1) {};
		\node[vertex] (2) [below right of=1] {};
		\node[vertex] (3) [above right of=2] {};
		\node[vertex] (4) [below of=2] {};
		\node[vertex] (5) [right of=4] {};
		\path [-, thick]
			(1) edge (2)
			(2) edge (3)
			(2) edge (4)
			(4) edge (5);

		\path (4.south)+(0,-1) node [text centered, text width=7em] (graph_text) {Communication Graph $\mathcal{G}$};

		\node [draw, rounded corners, thick] (graph) [fit=(1) (2) (3) (4) (5) (graph_text)] {};

		\draw[->, thick] (plant.30) -- node [above] {\large $y_1$} ++(0.7,0) |- (ob1.west);
		\draw[->, thick] (plant.-30) -- node [below] {\large $y_m$} ++(0.7,0) |- (obm.west);
		\path (plant.east)+(0.3, 0.1) node (ymid) {\LARGE $\vdots$};

		\draw [->, thick] (ob1.5) --  (graph.west |- ob1.5);
		\draw [<-, thick] (ob1.-5) --  (graph.west |- ob1.-5);

		\draw [->, thick] (obm.5) --  (graph.west |- obm.5);
		\draw [<-, thick] (obm.-5) --  (graph.west |- obm.-5);

		\node [rounded corners, draw=black!50, dashed, thick] (do) [fit=(ob1) (obm) (graph)] {};
		\path (do.south)+(0,-.5) node (dist_ob) {Distributed Observer};
	\end{tikzpicture}
	\caption {A framework for distributed state estimation.}
	\label{figure_framework}
\end{figure}
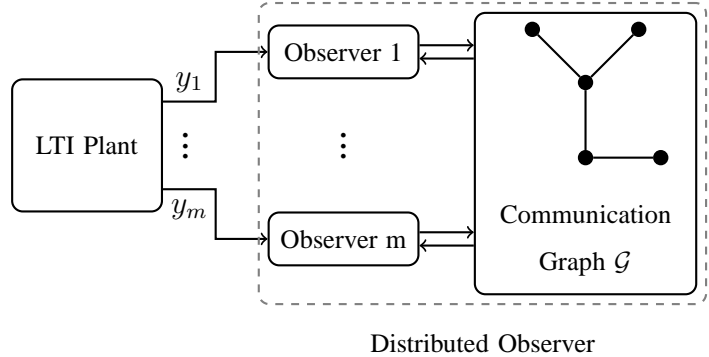

\textbf{Our main goals} are \textbf{(i)} given a plant \eqref{eq_LTI_plant} and a graph $\mathcal{G}$, to determine necessary and sufficient conditions for the existence of a LTI distributed observer that achieves omniscience, and \textbf{(ii)} to devise a method to obtain an omniscience-achieving solution, when one exists. \textbf{The main technical challenges} are that \textbf{(i)} the observers observe only a portion of the output of the plant and \textbf{(ii)} information exchange among the observers is constrained by the pre-selected graph $\mathcal{G}$.

\subsection{Motivation}
As will be specified in Section \ref{sec_formulation}, the class of update rules adopted in this work is distributed and linear. Some advantages of this class are briefly discussed in this subsection.

\subsubsection{Centralized vs Distributed}
For the sake of argument, we consider the following approach and call it \textit{centralized}: Suppose that under the same configuration as in Fig. \ref{figure_framework}, each observer would transmit its local measurement (the observed portion of the output of the plant) to its neighbors and, at the same time, would relay local measurements received from neighboring observers\footnote{This scheme is different from our approach as the observers in our framework exchange state estimates instead of measurements.}. If the underlying communication graph is \textit{well-connected}, then each observer would eventually receive enough information to estimate the entire state of the plant. In addition, if the states of the observers were not exchanged, then the dynamics of the observers would be  decoupled, and the design of each observer could be done by a (standard) centralized method.

However, this simple approach may not be suitable for implementation due to the profuse need for memory and communication resources. In particular, the centralized approach would require each observer to store its past estimates or measurements in its memory to account for delays incurred when exchanging information across multiple hops\footnote{This problem can be viewed as state estimation with delayed measurements. The reader is referred to \cite{gopalakrishnan2011_j_process_control} for a concise overview of existing (Kalman filter-based) methods.}. Moreover, this approach would not scale well because information transmission requirements would increase exponentially with the number of observers.

\subsubsection{Linear vs Nonlinear}
In stochastic control, it is well known that nonlinear controllers may outperform linear ones in some optimality criterion \cite{witsenhausen1968_siam_jc}. As an estimation problem can be formulated as an associated control problem, the same logic would hold for optimal estimation problems. However, in what regards to stability, we show that the proposed class of distributed observers performs equally well as nonlinear ones.

Robustness is an essential issue in feedback control problems, and design of robust control laws is of particular interesting, e.g., $H_2/H_{\infty}$ optimal control. There are abundant mathematical theories and computational algorithms for analysis and synthesis of LTI feedback systems with respect to certain robustness criteria \cite{dullerud2000_springer}. As the proposed class of distributed observers is linear and time-invariant, one would benefit from existing schemes in robust control literature in determining parameters for the observers.

\subsection{Contribution of This Work}
In order to achieve the goals, this paper focuses on the following two contributions: \textbf{(i)} We propose a parametrized class of distributed observers within which information exchange conforms to a pre-specified directed communication graph $\mathcal{G}$\footnote{Even though various forms of distributed observers are proposed in literature, the class of LTI distributed observers adopted in this paper, which is specified in Section \ref{sec_formulation}, is broader.}. \textbf{(ii)} We find necessary and sufficient conditions for the existence of parameters for a distributed observer in the aforementioned class that achieves omniscience asymptotically.

A detailed analysis is given in Section \ref{sec_proof_main_thm}, and hinges on the fact that omniscience for the proposed class of distributed observers can be cast as the stabilization of an associated LTI plant via fully decentralized control.

Using this analogy, we show that the existence of an omniscience-achieving distributed observer depends only on the detectability of the subsystems of the plant associated with the strong components (maximal strongly connected subgraphs) of $\mathcal{G}$. It follows from our analysis that if there are no omniscience-achieving solutions in the proposed class then omniscience cannot be attained by any other scheme -- including nonlinear and time-varying ones -- that is subject to the same graph.

\subsection{Paper Organization}
In Section \ref{sec_formulation} we define a parametrized class of distributed observers used throughout the paper, and we provide a comparative review of existing work. Section \ref{sec_main_result} gives the main result of this paper, which states the necessary and sufficient conditions for the existence of parameters for an omniscience-achieving distributed observer in the proposed class. An application for stabilization via distributed control is discussed in Section~\ref{sec_applications}. The detailed proof of the main result is provided in Section \ref{sec_proof_main_thm}. Section \ref{sec_conclusion} ends the paper with conclusions.

\section{Problem Formulation} \label{sec_formulation}
\subsection{Notation}
The following is the notation used throughout this paper.

\begin{tabularx} {\linewidth} {l X}
	$m$ & The number of subsystems as in \eqref{eq_LTI_plant}. \\
	$\mathbb{V}, \mathbb{E}$ & A vertex set defined as $\mathbb{V} \overset{def}{=} \{1, \cdots, m\}$ and an edge set $\mathbb{E} \subseteq \mathbb{V} \times \mathbb{V}$, respectively. \\
	$\mathcal{G} = \left( \mathbb{V}, \mathbb{E} \right)$ & A graph formed by the vertex set $\mathbb{V}$ and edge set $\mathbb{E}$\footnote{For notational convenience, we assume that no vertices of $\mathcal{G}$ have a loop, i.e., $(i,i) \notin \overline{\mathbb{E}}$, unless otherwise specified.}. \\
	$I_n$ & The $n$-dimensional identity matrix\\
	$\otimes$ & Kronecker product \\
	$sp(M)$ & The set of all eigenvalues of a matrix $M$, given by 
		
		$sp(M) = \{\lambda \in \mathbb{C} \mid det(M - \lambda I)=0\}$. \\
	$\Lambda_U(M)$ & The set of all unstable eigenvalues of a matrix $M$, given by 
	
		$\Lambda_U(M) = \{\lambda \in sp(M) \mid |\lambda| \geq 1\}$. \\
%
	$diag \left (M_1, \cdots, M_m \right )$ & For a set $\{M_1, \cdots, M_m\}$ of matrices, we define: 
	
		$diag \left (M_1, \cdots, M_m \right ) \overset{def}{=} \begin{pmatrix} M_{1} & \cdots & \mathbf{0} \\ \vdots & \ddots & \vdots \\ \mathbf{0} & \cdots & M_{m} \end{pmatrix}$. \\
	$W = (w_{ij})_{i,j \in \mathbb{V}}$ & For the set $\mathbb{V}$, $W = (w_{ij})_{i,j \in \mathbb{V}}$ is a matrix whose $i,j$-th element is $w_{ij}$. \\
	$v = (v_{i})_{i \in \mathbb{V}}$ & For the set $\mathbb{V}$, $v = (v_{i})_{i \in \mathbb{V}}$ is a vector whose $i$-th element is $v_{i}$.\\
	$B_{\mathbb{J}}, H_{\mathbb{J}}$ & For a set $\mathbb{J} = \{j_1, \cdots, j_p\} \subseteq  \{1, \cdots, m \}$ and matrices 
 $B$ and $H$ where $B=~\begin{pmatrix} B_1 & \cdots & B_m \end{pmatrix}$ and $H = \begin{pmatrix} H_1^T & \cdots & H_m^T \end{pmatrix}^T$, we define:
	
	$B_{\mathbb{J}} \overset{def}{=} \begin{pmatrix} B_{j_1} & \cdots & B_{j_p} \end{pmatrix}$ and $H_{\mathbb{J}} \overset{def}{=} \begin{pmatrix} H_{j_1}^T & \cdots & H_{j_p}^T \end{pmatrix}^T$. \\
\end{tabularx}

\subsection{The Class of LTI Distributed Observers and Main Problem} \
We consider that a LTI plant \eqref{eq_LTI_plant} and a directed communication graph $\mathcal{G}=\left( \mathbb{V}, \mathbb{E} \right)$ are given. Each vertex $i$ in $\mathbb{V}$ is associated with observer $i$, which assesses $y_i(k)=H_ix(k)$. We adopt the convention that $(i,j) \in \mathbb{E}$ if information can be transferred from observer $i$ to observer $j$. The neighborhood of observer $i$, denoted as $\mathbb{N}_i$, is a subset of $\mathbb{V}$ that contains $i$ and all other vertices with an outgoing edge towards $i$. Essentially, the elements of $\mathbb{N}_i$ represent the observers that can transmit information to observer $i$.

In this paper, we adopt the parametrized class of distributed observers inspired on \cite{park2012_ieee_cdc}, where each observer updates its state according to the following state-space equation:
\begin{equation} \label{eq_distributed_observer}
	\begin{split}
		\hat{x}_i (k+1) &= A \sum_{j \in \mathbb{N}_i} \mathbf{w}_{ij} \underbrace{\hat{x}_j (k)}_\text{state estimate} + \mathbf{K}_{i} \underbrace{ \left( y_i(k) - H_i \hat{x}_i (k)\right)}_\text{measurement innovation} + \mathbf{P}_{i} \underbrace{z_i (k)}_\text{augmented state}, ~ i \in \mathbb{V}\\
		z_i(k+1) &= \mathbf{Q}_{i} \left(y_i (k) - H_i \hat{x}_i (k) \right) + \mathbf{S}_{i} z_i (k)
	\end{split}
\end{equation}
where $A$ and $H_i$ are given in \eqref{eq_LTI_plant}, and $\mathbf{w}_{ij} \in \mathbb{R}$, $\mathbf{K}_{i} \in \mathbb{R}^{n \times r_{i}}$, $\mathbf{P}_{i} \in \mathbb{R}^{n \times \mu_i}$, $\mathbf{Q}_{i} \in \mathbb{R}^{\mu_i \times r_{i}}$, $\mathbf{S}_{i} \in \mathbb{R}^{\mu_i \times \mu_i}$ are the design parameters and $\mu_i$ is the dimension of the augmented state $z_i(k)$. We also refer to $\left\{ \mathbf{K}_i, \mathbf{P}_i, \mathbf{Q}_i, \mathbf{S}_i \right\}_{i \in \mathbb{V}}$ as gain matrices and $\mathbf{W} = \left( \mathbf{w}_{ij} \right)_{i,j \in \mathbb{V}}$ as a weight matrix that must satisfy $\sum_{j \in \mathbb{N}_i} \mathbf{w}_{ij} = 1$ for all $i \in\mathbb{V}$\footnote{We use bold font to represent the parameters to be designed.}. It follows from \eqref{eq_distributed_observer} that observer $i$ uses state estimates from within its neighborhood which implies that communication is distributed.

The following definition of omniscience-achieving parameters will be used throughout the paper.

\begin{mydef} [\textbf{Omniscience-achieving Parameters}] \label{def_omniscience_param}
	Consider a LTI plant with state $x(k)$ and a distributed observer with state estimates $\{\hat{x}_{i}(k)\}_{i \in \mathbb{V}}$ computed according to \eqref{eq_distributed_observer}. Any parameters $\mathbf{W}$ and $\left\{ \mathbf{K}_i, \mathbf{P}_i, \mathbf{Q}_i, \mathbf{S}_i \right\}_{i \in \mathbb{V}}$ of \eqref{eq_distributed_observer} are referred to as \textit{omniscience-achieving} if the resultant distributed observer achieves omniscience.
\end{mydef}

The following is the main problem addressed in this paper.

\begin{center}
	\line(1,0){450}
\end{center}

\textbf{\textit{Problem:}} Given a LTI plant \eqref{eq_LTI_plant} and a graph $\mathcal{G}$, determine necessary and sufficient conditions for the existence of a weight matrix $\mathbf{W} = \left( \mathbf{w}_{ij} \right)_{i,j \in \mathbb{V}}$ and gain matrices $\left\{ \mathbf{K}_i, \mathbf{P}_i, \mathbf{Q}_i, \mathbf{S}_i \right\}_{i \in \mathbb{V}}$ in \eqref{eq_distributed_observer} such that the corresponding distributed observer achieves omniscience asymptotically.
\begin{center}
	\line(1,0){450}
\end{center}

\subsection{Comparative Review of Related Work}
In \cite{olfati-saber2005_ieee_cdc_ecc, olfati-saber07_ieee_cdc}, the author introduced an algorithmic approach for distributed state estimation. The proposed method consists of a state estimation component, which is rooted on the Kalman filter, and a data fusion component, which utilizes a consensus algorithm \cite{olfati-saber2007_IEEEproceedings}. The performance of this approach is studied in \cite{carli2008_ieee_j_sac}, while its stability properties are reported in \cite{olfati-saber09_ieee_cdc, kamgarpour2008_ieee_cdc, khan2011_acc}.

Investigations of various distributed estimation schemes, which essentially have a similar structure as those in \cite{olfati-saber2005_ieee_cdc_ecc, olfati-saber07_ieee_cdc}, are then followed. The authors of \cite{matei2012_automatica} proposed a consensus-based linear observer and devised a method to obtain sub-optimal gain parameters. In \cite{khan2011_arxiv}, a consensus-based linear observer, which has a similar structure as one described in  \cite{matei2012_automatica}, is proposed where gain parameters are determined depending on the measurement matrix of the plant and the Laplacian matrix of the underlying communication graph. Other interesting approaches are reported in \cite{alriksson2006_mtns, khan2011_ieee_cdc, bai2011_acc}.

To achieve stability of state estimation, some of the existing distributed estimation algorithms require \textbf{(i)} strong observability conditions \cite{olfati-saber09_ieee_cdc}, \textbf{(ii)} multiple data fusion steps between two consecutive estimation steps \cite{kamgarpour2008_ieee_cdc, khan2011_acc}, which imposes a two-time-scale structure, or \textbf{(iii)} the verification of algebraic constraints \cite{bai2011_acc}, which is a stronger condition than the one presented in this paper.

\textbf{Comparison with prior publications by the authors:} The introduction of augmented states as in \eqref{eq_distributed_observer} was proposed in \cite{park2012_acc}, where we also provided sufficient conditions for the existence of omniscience-achieving parameters. In \cite{park2012_ieee_cdc} we developed necessary and sufficient conditions for the existence of omniscience-achieving gain matrices for the case where $\mathbf{W}$ is a pre-selected symmetric matrix. \underline{This paper \textit{extends} and \textit{unifies} our prior results in the following way}: We consider \textit{directed communication graphs}, which allows for asymmetric $\mathbf{W}$, and investigate necessary and sufficient conditions for the existence of omniscience-achieving schemes for which $\mathbf{W} =~\left( \mathbf{w}_{ij} \right)_{i,j \in \mathbb{V}}$ and $\left\{ \mathbf{K}_i, \mathbf{P}_i, \mathbf{Q}_i, \mathbf{S}_i \right\}_{i \in \mathbb{V}}$ in \eqref{eq_distributed_observer} are parameters that must be designed jointly.

\section{Main Result} \label{sec_main_result}
In this section, we present our main result and an example. We start with the following Definition of a source component of a graph.

\begin{mydef} \label{def_source_comp}
	Given a directed graph $\mathcal{G}= \left( \mathbb{V}, \mathbb{E} \right)$, a strongly connected component $\left( \mathbb{V}_c, \mathbb{E}_c \right)$ of $\mathcal{G}$ is said to be a \textit{source component} if there is no edge from $\mathbb{V} \setminus \mathbb{V}_c$ to $\mathbb{V}_c$ in $\mathcal{G}$.
\end{mydef}

The following is our main Theorem.

\begin{center}
	\line(1,0){450}
\end{center}
\begin{thm} [\textbf{Detectability Condition for Omniscience}] \label{thm_main}
	Suppose that the communication graph $\mathcal{G}= \left( \mathbb{V}, \mathbb{E} \right)$ is pre-selected, that the plant is given as in \eqref{eq_LTI_plant}, and that the following assumptions hold:
	\begin{itemize}
		\item[\textbf{(i)}] There are $s$ source components of $\mathcal{G}$, which are represented as $\left\{ \left( \mathbb{V}_i, \mathbb{E}_i \right) \right\}_{i \in \{1, \cdots, s\}}$.
		\item[\textbf{(ii)}] Each source component $i$ has an associated subsystem given by the pair $\left( A, H_{\mathbb{V}_i} \right)$.
	\end{itemize}
	There exist a choice of omniscience-achieving parameters $\mathbf{W}$ and $\{\mathbf{K}_i, \mathbf{P}_i, \mathbf{Q}_i, \mathbf{S}_i\}_{i \in \mathbb{V}}$ if and only if all subsystems $\left( A, H_{\mathbb{V}_i}\right)$ for $i \in \{1, \cdots, s\}$ are detectable.
\end{thm}
\begin{center}
	\line(1,0){450}
\end{center}

\begin{remark} \label{remark_param_computation}
	As will be discussed in Section \ref{sec_proof_main_thm}, once the detectability condition of Theorem~\ref{thm_main} is satisfied, under a proper choice of a weight matrix $\mathbf{W}$ (see the proof of Theorem \ref{thm_main} in Subsection~\ref{subsection_proof_thm_main}), we can compute omniscience-achieving gain matrices $\{\mathbf{K}_i, \mathbf{P}_i, \mathbf{Q}_i, \mathbf{S}_i\}_{i \in \mathbb{V}}$ via the methods proposed in \cite{wang1973_tac, davison1990_tac}.
\end{remark}

\tikzstyle{main node}=[circle, draw, text centered, thick]

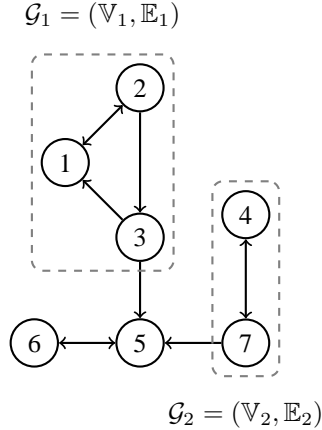
\begin{figure} [t]
	\centering
	\begin{tikzpicture} [node distance=40]
		\node[main node](1) {1};
		\node[main node](2) [above right of=1] {2};
		\node[main node](3) [below right of=1] {3};
		\node[main node](5) [below of=3] {5};
		\node[main node](6) [left of=5] {6};
		\node[main node](7) [right of=5] {7};
		\node[main node](4) [above=3em of 7] {4};

		\path [<->, thick]
			(1) edge (2)
			(4) edge (7)
			(5) edge (6);
			
		\path [->, thick]
			(2) edge (3)
			(3) edge (1)
			(3) edge (5)
			(7) edge (5);
			
		\node [rounded corners, draw=black!50, dashed, thick] (G1) [fit=(1) (2) (3)] {};
		\path (G1.north)+(0,.5) node (G1_eq) {$\mathcal{G}_1 = \left( \mathbb{V}_1, \mathbb{E}_1 \right)$};
		\node [rounded corners, draw=black!50, dashed, thick] (G2) [fit=(4) (7)] {};
		\path (G2.south)+(0,-.5) node (G2_eq) {$\mathcal{G}_2 = \left( \mathbb{V}_2, \mathbb{E}_2 \right)$};
	\end{tikzpicture}
	\caption {A communication graph $\mathcal{G}$ and its source components $\mathcal{G}_1$ and $\mathcal{G}_2$ for Example.}
	\label{figure_ex}
\end{figure}

\textit{\textbf{Example:}}	
	Consider the communication graph $\mathcal{G} = \left( \mathbb{V}, \mathbb{E} \right)$ depicted in Fig. \ref{figure_ex} and a LTI plant~\eqref{eq_LTI_plant} with $m=7$. By Definition \ref{def_source_comp}, we identify that $\mathcal{G}_1$ and $\mathcal{G}_2$ are the only source components of $\mathcal{G}$. Therefore, by Theorem \ref{thm_main}, we conclude that omniscience can be achieved if and only if $\left( A, H_{\mathbb{V}_1} \right)$ and $\left( A, H_{\mathbb{V}_2} \right)$ are both detectable.

\subsection{Additional properties and facts about the proposed class of distributed observers} \label{subsection_properties}
\begin{itemize}
\item[\textbf{(i)}] In the presence of process and measurement noises with finite second moment that enter linearly in \eqref{eq_LTI_plant}, the estimation error of our class of distributed observers has finite second moment.

\item[\textbf{(ii)}] Information is exchanged only among neighboring observers, and for achieving omniscience, it is sufficient to exchange local state estimates whose dimensions are equal to the order of the plant. 

\item[\textbf{(iii)}] If the detectability condition of Theorem \ref{thm_main} fails then there are neither omniscience-achieving parameters for \eqref{eq_distributed_observer}, nor any other nonlinear or time-varying omniscience-achieving scheme subject to the same communication graph.

\item[\textbf{(iv)}] Even though the optimization of the weight and gains, for instance, with respect to $H_{\infty}$ optimality criterion may be non-convex due to the sparse structure imposed on them, one may use a nonsmooth $H_{\infty}$ synthesis tool \cite{apkarian2006_ieee_tac, gahinet2011_ifac}, which is readily available in decentralized control literature, to obtain locally optimal solutions.

\item[\textbf{(v)}] We do not discuss the order of the observers, particularly the dimension of the augmented state $z_i$ as this issue has been explored in output feedback stabilization. For detailed discussion, the reader is referred to \cite{rosenthal1996_ieee_tac} and references therein.
\end{itemize}

\section{Application to the Design of Distributed Controllers} \label{sec_applications}
Consider a graph $\mathcal{G} = \left( \mathbb{V}, \mathbb{E} \right)$ and the following LTI plant with state $x(k)$, output vector $y(k)$, and inputs $\left\{u_i(k)\right\}_{i \in \mathbb{V}}$.
\begin{equation} \label{eq_LTI_plant_input}
	\begin{split}
		x(k+1) &= Ax(k) + \sum_{i \in \mathbb{V}} B_iu_i(k) \\
		y_i(k) &= H_i x(k), ~ i \in \mathbb{V}
	\end{split}
\end{equation}

In this section, we consider a distributed control problem as an application of the proposed estimation scheme. We focus on designing $m$ LTI controllers in which information exchange conforms with $\mathcal{G}$ and each controller has the following state-space representation:
\begin{equation} \label{eq_dist_controller}
	\begin{split}
		\xi_i (k+1) &= \sum_{j \in \mathbb{N}_i} \mathbf{S}_j^c \xi_j (k) + \mathbf{Q}_i^c y_i(k), ~ i \in \mathbb{V} \\
		u_i(k) &= \sum_{j \in \mathbb{N}_i} \mathbf{P}_j^c \xi_j (k) + \mathbf{K}_i^c y_i (k)
	\end{split}
\end{equation}
where $\xi_i$ is the internal state of controller $i$. We refer to $\mathcal{G}$ and the collection of all controllers as a \textit{distributed controller}.

\textbf{The goal} is to determine conditions for the existence of a distributed controller that stabilizes the plant \eqref{eq_LTI_plant_input} and to devise a method to compute $\left\{ \mathbf{K}_i^c, \mathbf{P}_i^c, \mathbf{Q}_i^c, \mathbf{S}_i^c \right\}_{i \in \mathbb{V}}$ if one exists\footnote{In \cite{wang1973_tac, davison1990_tac}, a problem of stabilizing a LTI plant via completely decoupled controllers, i.e., $\mathbb{N}_i = \{i\}$ for all $i \in \mathbb{V}$, is studied. In recent work \cite{pajic2011_ieee_cdc-ecc}, an idea of adopting a Wireless Control Network (WCN) is proposed where the WCN is a LTI system that bridges the plant and decoupled controllers. }. To do so, we will make use of the distributed observer described in Section \ref{sec_formulation}.

The following Proposition specifies sufficient conditions for the existence of a stabilizing distributed controller, where the computation of the parameters $\left\{ \mathbf{K}_i^c, \mathbf{P}_i^c, \mathbf{Q}_i^c, \mathbf{S}_i^c \right\}_{i \in \mathbb{V}}$ is  described in the proof of the Proposition (see Subsection \ref{sec_proof_prop_dist_cont}).

\begin{prop} \label{prop_dist_cont}
	Let a graph $\mathcal{G} = \left( \mathbb{V}, \mathbb{E} \right)$ and a LTI plant \eqref{eq_LTI_plant_input} be given. Suppose the following assumptions hold:
	\begin{itemize}
		\item[\textbf{(i)}] The plant is stabilizable.
		\item[\textbf{(ii)}] The graph $\mathcal{G}$ and the pair $\left(A, H \right)$ satisfy the detectability condition of Theorem \ref{thm_main}.
	\end{itemize}
	There exists a distributed controller \eqref{eq_dist_controller} that stabilizes the plant.
\end{prop}

\begin{remark}
	Notice that the aforementioned controllers share internal states within neighborhood defined by $\mathcal{G}$. We argue that  this scheme performs better than one that shares local measurements. For the sake of comparison, we consider controllers of the following form. Since they are sharing local measurements $y_i$ within neighborhood, we refer to them as the measurement-sharing controllers.
	\begin{equation} \label{eq_decent_controller}
		\begin{split}
			\xi_i (k+1) &= \mathbf{S}_i^c \xi_i (k) + \sum_{j \in \mathbb{N}_i} \mathbf{Q}_j^c y_j(k), ~ i \in \mathbb{V} \\
			u_i(k) &= \mathbf{P}_i^c \xi_i (k) + \sum_{j \in \mathbb{N}_i} \mathbf{K}_j^c y_j (k)
		\end{split}
	\end{equation}
	It can be verified that the assumptions \textbf{(i)} and \textbf{(ii)} of Proposition \ref{prop_dist_cont} are not sufficient for the existence of the measurement-sharing controllers \eqref{eq_decent_controller} which stabilize the plant \eqref{eq_LTI_plant_input} (see Corollary~2 of \cite{lavaei2008_automatica} for the stabilizability condition for a LTI plant via the measurement-sharing controllers)\footnote{In what regards to stability, the state-sharing controllers \eqref{eq_dist_controller} outperform the measurement-sharing controllers \eqref{eq_decent_controller} in the following sense: Under the same information exchange constraint if a plant can be stabilized by the measurement-sharing controllers than it can always be stabilized by the state-sharing controllers, but not vice versa. We omit the detail for brevity.}.
\end{remark}

To prove Proposition \ref{prop_dist_cont}, we consider a set of controllers governed by the following state-space equation. Notice that this is a special choice of \eqref{eq_dist_controller}.


\begin{equation} \label{eq_dist_controller_particular}
	\begin{split}
		\begin{pmatrix} \hat{x}_i(k+1) \\  z_i(k+1) \\ w_i(k+1) \end{pmatrix} 
			&= \begin{pmatrix} \sum_{j \in \mathbb{N}_i} \mathbf{w}_{ij} A \hat{x}_j(k) - \mathbf{K}_i \left( y_i(k) - H_i \hat{x}_i(k) \right) + \mathbf{P}_i z_i(k) \\
				-\mathbf{Q}_i \left( y_i(k) - H_i\hat{x}_i(k) \right) + \mathbf{S}_i z_i(k) \\
				\mathbf{Q}_i^d \hat{x}_i(k) + \mathbf{S}_i^d w_i(k)
		 \end{pmatrix}, ~ i \in \mathbb{V} \\
		u_i(k) &= \mathbf{K}_i^d \hat{x}_i(k) + \mathbf{P}_i^d w_i(k)
	\end{split}
\end{equation}

In what follows, we first consider a choice of $\mathbf{W} = \left( \mathbf{w}_{ij} \right)_{i,j \in \mathbb{V}}$ and $\left\{ \mathbf{K}_i, \mathbf{P}_i, \mathbf{Q}_i, \mathbf{S}_i\right\}_{i \in \mathbb{V}}$ (Step~I), and we consider a choice of $\left\{ \mathbf{K}^{d}_i, \mathbf{P}^{d}_i, \mathbf{Q}^{d}_i, \mathbf{S}^{d}_i\right\}_{i \in \mathbb{V}}$ (Step II). The proof of Proposition~\ref{prop_dist_cont} is then followed.

\subsection{Step I}
Consider the following LTI system with state $\begin{pmatrix} x^T(k) & \hat{x}^T(k) & z^T(k) \end{pmatrix}^T$, output vector $\hat{x}_i(k)$, and inputs $\{u_j(k)\}_{j \in \mathbb{V}}$:

\begin{align} \label{eq_plant_observer}
	\begin{split}
		\begin{pmatrix} x(k+1) \\ \hat{x}(k+1) \\ z(k+1) \end{pmatrix}
		&= \begin{pmatrix}
			A & \mathbf{0} & \mathbf{0} \\ 
			\overline{\mathbf{K}} ~ \overline{H} \left( \mathbf{1} \otimes I_n \right) & \mathbf{W} \otimes A - \overline{\mathbf{K}} ~ \overline{H} & \overline{\mathbf{P}} \\ 
			\overline{\mathbf{Q}} ~ \overline{H} \left( \mathbf{1} \otimes I_n \right) & - \overline{\mathbf{Q}} ~ \overline{H} & \overline{\mathbf{S}}
			\end{pmatrix} \begin{pmatrix} x(k) \\ \hat{x}(k) \\ z(k) \end{pmatrix} + \begin{pmatrix} \sum_{j \in \mathbb{V}} B_j u_j(k) \\ \mathbf{0} \\ \mathbf{0} \end{pmatrix} \\
	\hat{x}_i(k) &= \begin{pmatrix} \mathbf{0} & \cdots & I_n & \cdots & \mathbf{0} \end{pmatrix} \hat{x}(k), ~ i \in \mathbb{V}
	\end{split}
\end{align}
with
\begin{equation} \label{eq_plant_observer_02}
	\begin{split}
		&\hat{x} = \left( \hat{x}_1^T, \cdots, \hat{x}_m^T \right)^T, \quad z = \left( z_1^T, \cdots, z_m^T \right)^T \\
		&\overline{H} = \begin{pmatrix} \overline{H}_1^T & \cdots & \overline{H}_m^T \end{pmatrix}^T \text{ with } \overline{H}_i = e_i^T \otimes H_i \\
		& \mathbf{W} = \left( \mathbf{w}_{ij} \right)_{i,j \in \mathbb{V}} \\
		&\overline{\mathbf{K}} = diag \left(\mathbf{K}_1, \cdots, \mathbf{K}_m \right),
		\quad \overline{\mathbf{P}} = diag \left(\mathbf{P}_1, \cdots, \mathbf{P}_m \right) \\
		&\overline{\mathbf{Q}} = diag \left(\mathbf{Q}_1, \cdots, \mathbf{Q}_m \right),
		\quad \overline{\mathbf{S}} = diag \left(\mathbf{S}_1, \cdots, \mathbf{S}_m \right) \\
	\end{split}
\end{equation}
where $e_i$ is the $i$-th column of the $m$-dimensional identity matrix. Notice that \eqref{eq_plant_observer} is obtained by interconnecting the plant \eqref{eq_LTI_plant_input} and distributed observer \eqref{eq_distributed_observer}. We refer to this system as a \textit{plant/observer system}. 

The following Lemma states the stabilizability and detectability of the plant/observer system.
\begin{lemma} \label{lemma_dist_control}
	Let a graph $\mathcal{G} = \left( \mathbb{V}, \mathbb{E} \right)$ and a LTI plant \eqref{eq_LTI_plant_input} be given. Suppose that the assumptions \textbf{(i)} and \textbf{(ii)} of Proposition \ref{prop_dist_cont} hold. We can find $\mathbf{W}$, $\overline{\mathbf{K}}$, $\overline{\mathbf{P}}$, $\overline{\mathbf{Q}}$, $\overline{\mathbf{S}}$ in \eqref{eq_plant_observer} for which the resultant plant/observer system is both stabilizable and detectable for all $i \in \mathbb{V}$.
\end{lemma}
\begin{proof}
	First of all, notice that since \textbf{(ii)} of Proposition \ref{prop_dist_cont} holds, using Theorem \ref{thm_main} and Remark \ref{remark_param_computation}, one can find $\mathbf{W}$, $\overline{\mathbf{K}}$, $\overline{\mathbf{P}}$, $\overline{\mathbf{Q}}$, $\overline{\mathbf{S}}$ such that the matrix $\begin{pmatrix} \mathbf{W} \otimes A - \overline{\mathbf{K}} ~ \overline{H} & \overline{\mathbf{P}} \\ -\overline{\mathbf{Q}} ~ \overline{H} & \overline{\mathbf{S}} \end{pmatrix}$ is stable. Under this choice of $\mathbf{W}$, $\overline{\mathbf{K}}$, $\overline{\mathbf{P}}$, $\overline{\mathbf{Q}}$, $\overline{\mathbf{S}}$, we show the stabilizability and detectability of the resultant plant/observer system.

	The stabilizability directly follows from the stabilizability of the plant (\textbf{(i)} of Proposition \ref{prop_dist_cont}). The detectability can be proved by observing the fact that if $u_i =~0$ for all $i \in \mathbb{V}$, then it holds that $\hat{x}_i(k) \xrightarrow[k \rightarrow \infty] {} x(k)$ and $z_i(k) \xrightarrow[k \rightarrow \infty] {} 0$ for all $i \in \mathbb{V}$.

\end{proof}

\subsection{Step II}
Consider a set of decoupled controllers whose state-space representation is given as follows:
\begin{equation} \label{eq_decen_controller_app}
	\begin{split}
		w_i(k+1) &= \mathbf{S}_i^{d} w_i(k) + \mathbf{Q}_i^{d} \hat{x}_i(k), ~ i \in \mathbb{V} \\
		u_i(k) &= \mathbf{P}_i^{d} w_i(k) + \mathbf{K}_i^{d} \hat{x}_i(k)
	\end{split}
\end{equation}

Suppose that the controllers \eqref{eq_decen_controller_app} are applied to the plant/observer system \eqref{eq_plant_observer}. It can be verified that by the results of \cite{wang1973_tac, davison1990_tac}, if the plant/observer system \eqref{eq_plant_observer} is stabilizable and detectable for all $i \in \mathbb{V}$, one can find $\left\{ \mathbf{K}^{d}_i, \mathbf{P}^{d}_i, \mathbf{Q}^{d}_i, \mathbf{S}^{d}_i\right\}_{i \in \mathbb{V}}$ for which the resultant controllers stabilize the plant/observer system.

\subsection{Proof of Proposition \ref{prop_dist_cont}} \label{sec_proof_prop_dist_cont}
Suppose that the assumptions \textbf{(i)} and \textbf{(ii)} of Proposition \ref{prop_dist_cont} are satisfied. First we observe that an interconnection of \eqref{eq_LTI_plant_input} and \eqref{eq_dist_controller_particular} is equivalent to that of \eqref{eq_plant_observer} and \eqref{eq_decen_controller_app}. By following the procedures described in Step I and Step II, we can verify that with the certain choice of parameters, \eqref{eq_dist_controller_particular} stabilizes the plant \eqref{eq_LTI_plant_input}. Since \eqref{eq_dist_controller_particular} is a particular choice of \eqref{eq_dist_controller}, this proves the existence of a distributed controller \eqref{eq_dist_controller} that stabilizes the plant, which completes the proof of the Proposition.

\section{Proof of Main Theorem} \label{sec_proof_main_thm}

\tikzstyle{square_box}=[draw, rounded corners, text centered, very thick, minimum width=7em]

\def\boxdistance{-.6}
\def\blockdistance{3.5}

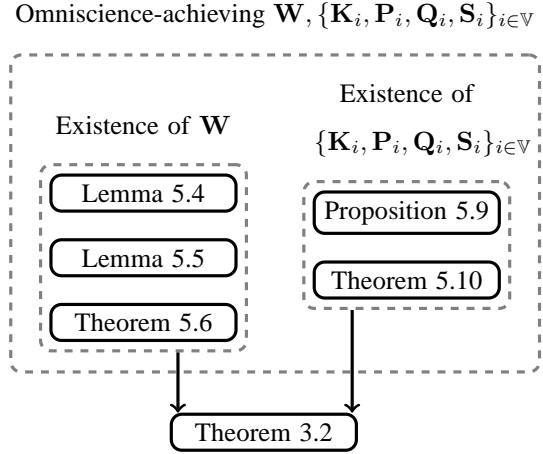
\begin{figure} [t]
	\centering
	\begin{tikzpicture} [node distance=40]
		\node[square_box](lem_1) {Lemma \ref{lemma_unique_rep}};
		\path (lem_1.south)+(0,\boxdistance) node [square_box] (lem_2) {Lemma \ref{lemma_eigen_rep}};
		\path (lem_2.south)+(0,\boxdistance) node [square_box] (thm_1) {Theorem \ref{thm_wlm}};
		\node [rounded corners, draw=black!50, dashed, very thick] (B1) [fit=(lem_1) (lem_2) (thm_1)] {};
		\path (B1.north)+(0,.5) node (descrip_1) {Existence of $\mathbf{W}$};
		
		\path (lem_1.south)+(\blockdistance,0) node [square_box] (prop_1) {Proposition \ref{prop_fixed_modes}};
		\path (prop_1.south)+(0,\boxdistance) node [square_box] (thm_2) {Theorem \ref{thm_decentralized_controller}};
		\node [rounded corners, draw=black!50, dashed, very thick] (B2) [fit=(prop_1) (thm_2)] {};
		\path (B2.north)+(0,.8) node [text centered, text width=7em]  (descrip_2) {Existence of $\{\mathbf{K}_i, \mathbf{P}_i, \mathbf{Q}_i, \mathbf{S}_i\}_{i \in \mathbb{V}}$};
		
		\node [rounded corners, draw=black!50, dashed, very thick, minimum width=20em, minimum height=12em] (B3) [fit=(descrip_1) (descrip_2) (B1) (B2)] {};
		\path (B3.north)+(0,.5) node [text centered]  (descrip_3) {Omniscience-achieving $\mathbf{W}, \{\mathbf{K}_i, \mathbf{P}_i, \mathbf{Q}_i, \mathbf{S}_i\}_{i \in \mathbb{V}}$};
		
		\path (thm_1.south)+(1.6,2*\boxdistance) node [square_box] (thm_3) {Theorem \ref{thm_main}};
		
		\draw [->, very thick] (B1.-70) --  (thm_3.north -| B1.-70);
		\draw [->, very thick] (B2.-130) --  (thm_3.north -| B2.-130);
		
	\end{tikzpicture}
	\caption {Precedence diagram for the proof of Theorem \ref{thm_main}}
	\label{figure_diagram}
\end{figure}

In this section, we present a two-part proof for Theorem \ref{thm_main}. The first part consists of  Lemma~\ref{lemma_unique_rep}, Lemma~\ref{lemma_eigen_rep}, and Theorem \ref{thm_wlm} that determine conditions for the existence of a suitable weight matrix $\mathbf{W}$ endowed with particular spectral properties. Given a suitable weight matrix, the second part, which consists of Proposition \ref{prop_fixed_modes} and Theorem~\ref{thm_decentralized_controller}, determines conditions for the existence of gain matrices $\{\mathbf{K}_i, \mathbf{P}_i, \mathbf{Q}_i, \mathbf{S}_i\}_{i \in \mathbb{V}}$ that, in conjunction with the given $\mathbf{W}$, are omniscience-achieving. The structure of the proof is outlined in the diagram of Fig. \ref{figure_diagram}.

\subsection{Useful Results on Weighted Laplacian Matrices}
\begin{mydef} \label{def_wlm}
	Consider a graph $\mathcal{G} = \left( \mathbb{V}, \mathbb{E} \right)$. A matrix $L = (l_{ij})_{i,j \in \mathbb{V}} \in~\mathbb{R}^{m \times m}$ is said to be a \textit{Weighted Laplacian Matrix (WLM)} of $\mathcal{G}$ if the following three properties hold:
	\begin{itemize}
		\item[\textbf{(i)}] If $(i,j) \notin \mathbb{E}$ then $l_{ji} = 0$ for $i \in \mathbb{V}$ and $j \in \mathbb{V} \setminus \{i\}$.
		\item[\textbf{(ii)}] If $(i,j) \in \mathbb{E}$ then $l_{ji} < 0$ for $i \in \mathbb{V}$ and $j \in \mathbb{V} \setminus \{i\}$.
		\item[\textbf{(iii)}] It holds that $\sum_{j \in \mathbb{V}} l_{ij} = 0$ for $i \in \mathbb{V}$.
	\end{itemize}
	For notational convenience, we define the following set of WLMs of $\mathcal{G}$:
$$\mathbb{L}(\mathcal{G}) \overset{def}{=} \{L \in \mathbb{R}^{m \times m} \mid L \text{ is a WLM of } \mathcal{G}\}$$
\end{mydef}

\begin{mydef} \label{def_rooted_tree}
	A directed graph $\mathcal{T} = \left( \mathbb{V}_\mathcal{T}, \mathbb{E}_\mathcal{T} \right)$ is said to be a \textit{rooted tree} if $\mathcal{T}$ has $\left( |\mathbb{V}_\mathcal{T}|-1 \right)$ edges and there exists a vertex $r \in \mathbb{V}_\mathcal{T}$, called a root of $\mathcal{T}$, such that for every $v \in \mathbb{V}_\mathcal{T} \setminus \{r\}$, there exists a directed path from root~$r$ to vertex $v$.
\end{mydef}

\begin{mydef} [\textbf{UEPP}] \label{def_uepp}
	Given square matrices $A$ and $B$, $A \otimes B$ is said to satisfy the so called \textit{Unique Eigenvalue Product Property (UEPP)} if every nonzero eigenvalue $\lambda$ of $A \otimes B$ can be uniquely expressed as a product $\lambda = \lambda_A \cdot \lambda_B$, where $\lambda_A$ and $\lambda_B$ are eigenvalues of $A$ and $B$, respectively\footnote{For an eigenvalue $\lambda \in sp(A \otimes B)$, let $\lambda_A, \lambda_A' \in sp(A)$ and $\lambda_B, \lambda_B' \in sp(B)$ for which $\lambda = \lambda_A \cdot \lambda_B = \lambda_A' \cdot \lambda_B'$. The eigenvalue $\lambda$ is said to be uniquely expressed as a product $\lambda=\lambda_A \cdot \lambda_B$ if it holds that $\lambda_A=\lambda_A'$ and $\lambda_B=\lambda_B'$.}.
\end{mydef}

\begin{lemma} \label{lemma_unique_rep}
	Given matrices $A \in \mathbb{R}^{n \times n}$ and $L \in \mathbb{R}^{m \times m}$, consider that $\mathbf{W}$ is of the form $\mathbf{W}=~I_m-~\alpha L$ where $\alpha$ is a positive real number. The following are true:
	\begin{itemize}
		\item[\textbf{(i)}] There is a positive real $\alpha$ such that $\mathbf{W} \otimes A$ satisfies the UEPP.
		\item[\textbf{(ii)}] If $L$ is a WLM of a graph, then for some positive $\alpha$, $\mathbf{W}=~I_m -~\alpha L$ becomes a stochastic matrix and $\mathbf{W} \otimes A$ satisfies the UEPP.
	\end{itemize}
\end{lemma}
\begin{proof}
	The proof is given in Appendix \ref{appen_lemma_proof}.
\end{proof}


\begin{lemma} \label{lemma_eigen_rep}
	Suppose that matrices $\mathbf{W} \in \mathbb{R}^{m \times m}$ and $A \in \mathbb{R}^{n \times n}$ are given where $\mathbf{W}$ has all simple eigenvalues\footnote{An eigenvalue of a matrix is simple if both the geometric and algebraic multiplicities of the eigenvalue are equal to $1$.}, and $\mathbf{W} \otimes A$ satisfies the UEPP. Each eigenvector $q$ of $\mathbf{W} \otimes A$ associated with $\lambda \in~sp(\mathbf{W} \otimes~A) \setminus~\{0\}$ can be written as a Kronecker product $q=v \otimes p$, where $v$ and $p$ are eigenvectors of $\mathbf{W}$ and $A$ (associated with $\lambda_{\mathbf{W}} \in sp(\mathbf{W})$ and $\lambda_A \in sp(A)$, respectively, for which $\lambda = \lambda_{\mathbf{W}} \cdot \lambda_{A}$).
\end{lemma}

\begin{proof}
	By the UEPP of $\mathbf{W} \otimes A$, there exists a unique pair of eigenvalues $\lambda_{\mathbf{W}} \in sp(\mathbf{W})$ and $\lambda_{A} \in sp(A)$ for which it holds that $\lambda = \lambda_{\mathbf{W}} \cdot \lambda_{A}$. Since $\mathbf{W}$ has all simple eigenvalues, the following equality can be shown (the proof is along the same lines as that of Lemma 3.1 in \cite{park2012_acc}):
	\begin{equation} \label{eq_lemma_eigen_rep_01}
		g_{\mathbf{W} \otimes A} (\lambda) = g_A (\lambda_A)
	\end{equation}
	where $g_{\mathbf{W} \otimes A} (\lambda)$ and $g_A (\lambda_A)$ are respectively the geometric multiplicities of $\lambda$ and $\lambda_A$.

	Notice that, since the eigenvalues of $\mathbf{W}$ are all simple, there exists a unique  eigenvector (unique up to a scale factor), say $v$, associated with $\lambda_{\mathbf{W}}$. Together with this fact, by \eqref{eq_lemma_eigen_rep_01}, it can be shown that an eigenvector $q$ of $\mathbf{W} \otimes A$ associated with $\lambda$ can be written as $q = v \otimes p$ where $p$ is an eigenvector of $A$ associated with $\lambda_{A}$. This proves the Lemma.
\end{proof}


\begin{thm} \label{thm_wlm}
	Let a strongly connected graph $\mathcal{G} = \left( \mathbb{V}, \mathbb{E} \right)$ be given.  Almost all elements of the set $\mathbb{L}(\mathcal{G})$ satisfy the following properties:
	\begin{itemize}
		\item[\textbf{(P1)}] All right and left eigenvectors have no zero entries.
		\item[\textbf{(P2)}] All eigenvalues are simple.
	\end{itemize}
\end{thm}
\begin{proof}
	Since the proof is lengthy and needs certain preliminary results on structured linear system theory, we provide a review of key properties of structured linear systems along with a  proof of Theorem \ref{thm_wlm} in Appendix \ref{appen_weighted_laplacian}.
\end{proof}

\subsection{Brief Introduction to Stabilization via Decentralized Control}
We start by reviewing certain classical results in decentralized control that will be used in the proof of Theorem~\ref{thm_main}. Of special relevance are the results in \cite{wang1973_tac, davison1990_tac} that show that the existence of a stabilizing decentralized controller for a LTI plant can be characterized using the notion of fixed modes \cite{anderson1981_automatica, davison1983_automatica}, which is analogous to the concept of uncontrollable or unobservable modes in classical centralized control problems.

In order to give various definitions and concepts needed throughout this section, we will analyze the effect of decentralized feedback on the following plant, which will also be used to introduce certain key concepts used throughout the paper:
\begin{equation} \label{eq_decen}
	\begin{split}
		\widetilde{x}(k+1) &= \widetilde{A} \widetilde{x}(k) + \sum_{i \in \mathbb{V}} \widetilde{B}_i \widetilde{u}_i(k) \\
		\widetilde{y}_i(k) &= \widetilde{H}_i \widetilde{x}(k), ~ i \in \mathbb{V}
	\end{split}
\end{equation}
where $\widetilde{x}(k) \in \mathbb{R}^{\widetilde{n}}$, $\widetilde{y}_i(k) \in \mathbb{R}^{\widetilde{r}_i}$, and $\widetilde{u}_i(k) \in \mathbb{R}^{\widetilde{p}_i}$ are the state, $i$-th output, and $i$-th control.

\begin{mydef}
	A given $\lambda \in \mathbb{C}$ is a fixed mode of \eqref{eq_decen} if it is an eigenvalue of $\widetilde{A}~+~\sum_{i \in \mathbb{V}}\widetilde{B}_i\widetilde{K}_i\widetilde{H}_i$ for all $\widetilde{K}_i \in \mathbb{R}^{\widetilde{p}_i \times \widetilde{r}_i}$.
\end{mydef}
\begin{remark}
	The fixed mode is an eigenvalue of the plant \eqref{eq_decen} which is invariant under output feedback $\widetilde{u}_i(k) =~\widetilde{K}_i \widetilde{y}_i(k),~ i \in \mathbb{V}$, where $\widetilde{K}_i \in \mathbb{R}^{\widetilde{p}_i \times \widetilde{r}_i}$.
\end{remark}
	
The fixed modes can be characterized by an algebraic rank test as described in the following Proposition.
\begin{prop} \cite{davison1983_automatica} \label{prop_fixed_modes}
	Consider that a LTI plant is given as in \eqref{eq_decen}. Let $\widetilde{B} = \begin{pmatrix} \widetilde{B}_1 & \cdots & \widetilde{B}_m \end{pmatrix}$ and $\widetilde{H} =~\begin{pmatrix} \widetilde{H}_1^T & \cdots & \widetilde{H}_m^T \end{pmatrix}^T$. A given $\lambda \in \mathbb{C}$ is a fixed mode of the plant if and only if there exists $\mathbb{J} \subseteq \mathbb{V}$ such that 
	\begin{equation}
		rank \begin{pmatrix} \widetilde{A} - \lambda I_{\widetilde{n}} & \widetilde{B}_{\mathbb{J}} \\ \widetilde{H}_{\mathbb{J}^{c}} & \mathbf{0} \end{pmatrix} <  \widetilde{n},
	\end{equation}
	where $\widetilde{n}$ is the dimension of the matrix $\widetilde{A}$ and $\mathbb{J}^{c} = \mathbb{V} \setminus \mathbb{J}$.
\end{prop}

\begin{thm} \cite{wang1973_tac} \label{thm_decentralized_controller}
	Given a LTI plant as in \eqref{eq_decen}, consider output feedback of the following form:
\begin{equation} \label{eq_decen_controller}
	\begin{split}
		z_i(k+1) &= \widetilde{S}_i z_i(k) + \widetilde{Q}_i \widetilde{y}_i(k), ~ i \in \mathbb{V} \\
		\widetilde{u}_i(k) &= \widetilde{P}_i z_i(k) + \widetilde{K}_i \widetilde{y}_i(k)
	\end{split}
\end{equation}
where $z_i(k) \in \mathbb{R}^{\mu_i}$ for some $\mu_i \in \mathbb{N} \cup \{0\}$. If every unstable fixed mode is located inside the unit circle in $\mathbb{C}$ then there exists a parameter choice $\left\{ \widetilde{K}_i, \widetilde{P}_i, \widetilde{Q}_i, \widetilde{S}_i \right\}_{i \in \mathbb{V}}$ for which the resultant closed-loop system is stable.
\end{thm}

\begin{remark}
By applying \eqref{eq_decen_controller} into \eqref{eq_decen}, we can write the overall state-space equation in the following compact form:
\begin{equation} \label{eq_decen_form}
	\begin{split}
		\begin{pmatrix} \widetilde{x}(k+1) \\ z(k+1) \end{pmatrix} = \begin{pmatrix} \widetilde{A} + \widetilde{B} \widetilde{K} \widetilde{H} & \widetilde{B} \widetilde{P} \\ \widetilde{Q} \widetilde{H} & \widetilde{S} \end{pmatrix} \begin{pmatrix} \widetilde{x}(k) \\ z(k) \end{pmatrix}
	\end{split}
\end{equation}
where 
\begin{equation*}
	\begin{split}
		&z = \left( z_1^T, \cdots, z_m^T \right)^T, \quad \widetilde{B} = \begin{pmatrix} \widetilde{B}_1 & \cdots & \widetilde{B}_m \end{pmatrix}, \quad \widetilde{H} = \begin{pmatrix} \widetilde{H}_1^T & \cdots & \widetilde{H}_m^T \end{pmatrix}^T \\
		&\widetilde{K} = diag\left( \widetilde{K}_1, \cdots, \widetilde{K}_m \right), \quad  \widetilde{P} = diag\left( \widetilde{P}_1, \cdots, \widetilde{P}_m \right) \\
		&\widetilde{Q} = diag\left( \widetilde{Q}_1, \cdots, \widetilde{Q}_m \right), \quad \widetilde{S} = diag\left( \widetilde{S}_1, \cdots, \widetilde{S}_m \right)
	\end{split}
\end{equation*}
\end{remark}

\subsection{Proof of Theorem \ref{thm_main}} \label{subsection_proof_thm_main}
To analyze the stability of the proposed estimation scheme, we group the estimation error and the augmented states of all observers to obtain an overall state-space representation as in \eqref{eq_decen_form}. This is useful because it can be used to show that finding omniscience-achieving parameters can be equivalently stated as finding a stabilizing decentralized controller for an associated LTI system. This idea, in conjunction with Theorem \ref{thm_decentralized_controller}, allows us to connect the absence of unstable fixed modes for an appropriate decentralized control system with the existence of an omniscience-achieving scheme.


We proceed by writing the error dynamics of \eqref{eq_distributed_observer} as follows\footnote{Notice that the error dynamics is completely decoupled from the state estimates under the condition $\sum_{j \in \mathbb{N}_i} \mathbf{w}_{ij}=1$ for all $i \in \mathbb{V}$.}:
\begin{equation} \label{eq_error_dynamics02}
	\begin{split}
		\widetilde{x}_i (k+1) &= A \sum_{j \in {\mathbb{N}}_i} \mathbf{w}_{ij} \widetilde{x}_j(k) - \mathbf{K}_{i} H_i \widetilde{x}_i(k) - \mathbf{P}_{i} z_i (k), ~ i \in \mathbb{V} \\
		z_i(k+1) &= \mathbf{Q}_{i} H_i \widetilde{x}_i (k) + \mathbf{S}_{i} z_i (k),
	\end{split}
\end{equation}
where $\widetilde{x}_i(k) \overset{def}{=} x(k) - \hat{x}_i(k)$. Furthermore, we can rewrite \eqref{eq_error_dynamics02} as follows:

\begin{equation} \label{eq_error_dynamics_col}
	\begin{split}
		\begin{pmatrix} \widetilde{x}(k+1) \\ z(k+1)\end{pmatrix} = \begin{pmatrix} \mathbf{W} \otimes A - \overline{B} ~ \overline{\mathbf{K}} ~ \overline{H} & - \overline{B} ~ \overline{\mathbf{P}} \\ \overline{\mathbf{Q}} ~ \overline{H} & \overline{\mathbf{S}} \end{pmatrix} \begin{pmatrix} \widetilde{x}(k) \\ z(k)\end{pmatrix}
	\end{split}
\end{equation}
with
\begin{equation} \label{eq_variable_def}
	\begin{split}
		&\widetilde{x} = \begin{pmatrix} \widetilde{x}_1^T & \cdots & \widetilde{x}_m^T \end{pmatrix}^T, \quad z = \begin{pmatrix} z_1^T & \cdots & z_m^T \end{pmatrix}^T \\
		&\overline{B} = \begin{pmatrix} \overline{B}_1 & \cdots & \overline{B}_m \end{pmatrix} \text{ with } \overline{B}_i = e_i \otimes I_n \\
		&\overline{H} = \begin{pmatrix} \overline{H}_1^T & \cdots & \overline{H}_m^T \end{pmatrix}^T \text{ with } \overline{H}_i = e_i^T \otimes H_i \\
		&\overline{\mathbf{K}} = diag \left(\mathbf{K}_1, \cdots, \mathbf{K}_m \right),
		\quad \overline{\mathbf{P}} = diag \left(\mathbf{P}_1, \cdots, \mathbf{P}_m \right) \\
		&\overline{\mathbf{Q}} = diag \left(\mathbf{Q}_1, \cdots, \mathbf{Q}_m \right),
		\quad \overline{\mathbf{S}} = diag \left(\mathbf{S}_1, \cdots, \mathbf{S}_m \right), \\
	\end{split}
\end{equation}
where $e_i$ is the $i$-th column of the $m$-dimensional identity matrix.

Notice that \eqref{eq_error_dynamics_col} can be viewed as the state-space representation of a closed-loop system obtained by applying decentralized output feedback, parametrized by $\left\{ \mathbf{K}_i, \mathbf{P}_i, \mathbf{Q}_i, \mathbf{S}_i \right\}_{i \in \mathbb{V}}$, to a LTI system, described by $(\mathbf{W} \otimes A, \overline{B}, \overline{H})$. Hence, we can write  \eqref{eq_error_dynamics_col} as in \eqref{eq_decen_form} by selecting $\widetilde{A} = \mathbf{W} \otimes A$, $\widetilde{B} = -\overline{B}$, $\widetilde{H} = \overline{H}$, $\widetilde{K} = \overline{\mathbf{K}}$, $\widetilde{P} = \overline{\mathbf{P}}$, $\widetilde{Q} = \overline{\mathbf{Q}}$, and $\widetilde{S} = \overline{\mathbf{S}}$. 

Based on the aforementioned relation, under the assumption that there are no unstable fixed modes in $(\mathbf{W}~\otimes~A, \overline{B}, \overline{H})$, we can apply Theorem \ref{thm_decentralized_controller} to conclude that we are ready to apply the design procedures proposed in \cite{wang1973_tac, davison1990_tac} to compute the gain matrices $\left\{ \mathbf{K}_i, \mathbf{P}_i, \mathbf{Q}_i, \mathbf{S}_i \right\}_{i \in \mathbb{V}}$ that, in conjunction with the given $\mathbf{W}$, are omniscience-achieving.

The following Lemma is used in the proof of Theorem \ref{thm_main}.
\begin{lemma} \label{lemma_weight_matrix}
	Let a matrix $A \in \mathbb{R}^{n \times n}$ and a directed graph $\mathcal{G} = \left( \mathbb{V}, \mathbb{E} \right)$ with two source components $\mathcal{G}_1 =~\left( \mathbb{V}_1, \mathbb{E}_1 \right)$ and $\mathcal{G}_2 = \left( \mathbb{V}_2, \mathbb{E}_2 \right)$ be given. There exists a matrix $\mathbf{W} \in \mathbb{R}^{m \times m}$ for which the following hold:
	\begin{itemize}
		\item[\textbf{(F0)}] $\mathbf{W}$ satisfies $\mathbf{W} \cdot \mathbf{1} = \mathbf{1}$ and has the following structure:
		\begin{equation} \label{eq_lemma_weight_matrix_01}
			\mathbf{W} = \begin{pmatrix} \mathbf{W}_{1} & \mathbf{0} & \mathbf{0} \\ \mathbf{0} & \mathbf{W}_{2} & \mathbf{0}  \\ \mathbf{W}_{31} & \mathbf{W}_{32} & \mathbf{W}_{33} \end{pmatrix}
		\end{equation}
where the sparsity patterns of $\mathbf{W}_1 \in \mathbb{R}^{|\mathbb{V}_1| \times |\mathbb{V}_1|}$, $\mathbf{W}_2 \in \mathbb{R}^{|\mathbb{V}_2| \times |\mathbb{V}_2|}$, and $\mathbf{W} \in \mathbb{R}^{|\mathbb{V}| \times |\mathbb{V}|}$ are consistent with $\mathcal{G}_1$, $\mathcal{G}_2$, and $\mathcal{G}$, respectively\footnote{The sparsity pattern of a matrix $\mathbf{W} = \left( \mathbf{w}_{ij} \right)_{i,j \in \mathbb{V}}$ is consistent with a graph $\mathcal{G} = \left( \mathbb{V}, \mathbb{E} \right)$ if $\mathbf{w}_{ij}=0$ for $(j,i) \notin \mathbb{E}$.}.

	
		\item[\textbf{(F1)}] All the right and left eigenvectors of $\mathbf{W}_{1}$ and $\mathbf{W}_{2}$ have no zero entries.
		\item[\textbf{(F2)}] All the eigenvalues of $\mathbf{W}_{1}$ and $\mathbf{W}_{2}$ are simple.
		\item[\textbf{(F3)}] For $i \in \{1, 2\}$, each eigenvector $q$ of $\mathbf{W}_i \otimes A$ associated with $\lambda \in \Lambda_U \left( \mathbf{W}_i \otimes A\right)$ can be written as a Kronecker product $q=v \otimes p$, where $v$ and $p$ are eigenvectors of $\mathbf{W}$ and $A$ (associated with $\lambda_{\mathbf{W}} \in sp\left( \mathbf{W} \right)$ and $\lambda_{A} \in \Lambda_U\left( A \right)$, respectively, for which $\lambda = \lambda_{\mathbf{W}} \cdot \lambda_A$).

		\item[\textbf{(F4)}] It holds that $\Lambda_U \left( \mathbf{W}_{33} \otimes A \right) = \emptyset$.
	\end{itemize}
\end{lemma}
\begin{proof}
	As we shall see later, our construction of $\mathbf{W}$ automatically guarantees \textbf{(F0)}. Thus, we will focus on showing the facts \textbf{(F1)}-\textbf{(F4)}.

	Consider $\mathbf{W}_i = I_{|\mathbb{V}_i|} - \alpha_i L_i$ for $i \in \{1, 2\}$, where $\alpha_i$ is a positive real number and $L_i \in \mathbb{L}(\mathcal{G}_i)$. By Theorem \ref{thm_wlm}, we can select $L_i$ that satisfies \textbf{(P1)} and \textbf{(P2)} of Theorem \ref{thm_wlm}. This choice of $L_i$ leads to $\mathbf{W}_i$ that satisfies \textbf{(F1)} and \textbf{(F2)}.
	
	According to Lemma \ref{lemma_unique_rep}, we can choose $\alpha_i$ such that $\mathbf{W}_i$ is stochastic and $\mathbf{W}_i \otimes A$ satisfies the UEPP (see Definition \ref{def_uepp}). Since $\mathbf{W}_i$ is stochastic, its eigenvalues lie in or inside the unit circle; hence, an eigenvalue $\lambda \in~\Lambda_U \left( \mathbf{W}_i \otimes A\right)$ can be written as $\lambda = \lambda_{\mathbf{W}_i} \cdot \lambda_A$ for some $\lambda_{\mathbf{W}_i} \in sp \left( \mathbf{W}_i \right)$ and $\lambda_A \in \Lambda_U \left( A \right)$. Along with \textbf{(F2)} and the fact that $\mathbf{W}_i \otimes A$ satisfies the UEPP, using Lemma~\ref{lemma_eigen_rep}, we can verify that \textbf{(F3)} holds.

	\textbf{Auxiliary fact:} To show \textbf{(F4)}, we claim that there exist $\mathbf{W}_{31}$,  $\mathbf{W}_{32}$, and $\mathbf{W}_{33}$ such that each row of $\begin{pmatrix} \mathbf{W}_{31} & \mathbf{W}_{32} & \mathbf{W}_{33} \end{pmatrix}$ sums to one and all the eigenvalues of $\mathbf{W}_{33}$ are arbitrarily small; hence, for this choice of $\mathbf{W}_{31}$,  $\mathbf{W}_{32}$, and $\mathbf{W}_{33}$, it holds that $\Lambda_U \left( \mathbf{W}_{33} \otimes A \right) = \emptyset$. This proves \textbf{(F4)}, and it remains to prove the claim.
	
	\textbf{Proof of the auxiliary fact:} Recall that, due to \textbf{(F0)}, $\mathbf{W}$ needs to be consistent with $\mathcal{G}$, which restricts the choice of the elements of $\mathbf{W}_{31}$, $\mathbf{W}_{32}$, and $\mathbf{W}_{33}$. Let $\mathbb{V}_3 = \mathbb{V} \setminus \left( \mathbb{V}_1 \cup \mathbb{V}_2 \right)$. Notice that $\mathbb{V}_3$ can be \textit{spanned} by a collection of disjoint rooted trees in which each tree is rooted at a vertex of $\mathbb{V}_3$.

Since vertices in $\mathbb{V}_3$ do not belong to any of the source components, $\mathcal{G}_1$ and $\mathcal{G}_2$, there is a directed path from at least one of the source components to each vertex in $\mathbb{V}_3$. Hence, we may assume that each of the disjoint rooted trees that span $\mathbb{V}_3$ has the root which is a neighbor of $\mathbb{V}_1$ or $\mathbb{V}_2$ in $\mathcal{G}$, i.e., there is an edge from $\mathbb{V}_1$ or $\mathbb{V}_2$ to the root (of each rooted tree) in $\mathcal{G}$.

	Notice that after a permutation if necessary, a matrix that is consistent with a rooted tree is lower triangular. For this reason,  we may assume that $\mathbf{W}_{33}$ is lower triangular. 
		
	Also note that since there is a directed path from $\mathbb{V}_1$ or $\mathbb{V}_2$ to every vertex in $\mathbb{V}_3$, at least two elements, which include one diagonal element of $\mathbf{W}_{33}$, of each row of $\begin{pmatrix} \mathbf{W}_{31} & \mathbf{W}_{32} & \mathbf{W}_{33} \end{pmatrix}$ can be chosen to be non-zero. By properly choosing the lower triangular elements of $\mathbf{W}_{33}$ and elements of $\mathbf{W}_{31}$ and $\mathbf{W}_{32}$, we can make the diagonal elements of $\mathbf{W}_{33}$ arbitrarily small and each row of $\begin{pmatrix} \mathbf{W}_{31} & \mathbf{W}_{32} & \mathbf{W}_{33} \end{pmatrix}$ sums to one. Since $\mathbf{W}_{33}$ is a lower triangular matrix with arbitrarily small diagonal elements, all the eigenvalues of $\mathbf{W}_{33}$ are arbitrarily small. This proves the claim.	
\end{proof}

\quad\textit{Proof of Theorem \ref{thm_main}:} 

	\textbf{Sufficiency:} If we can choose a matrix $\mathbf{W}$, whose sparsity pattern is consistent with $\mathcal{G}$ and which satisfies $\mathbf{W} \cdot \mathbf{1} = \mathbf{1}$, such that no unstable fixed modes exist in $\left( \mathbf{W} \otimes A, \overline{B}, \overline{H}\right)$, then, by Theorem \ref{thm_decentralized_controller}, it follows the existence of gain matrices $\{\mathbf{K}_i, \mathbf{P}_i, \mathbf{Q}_i, \mathbf{S}_i\}_{i \in \mathbb{V}}$ that, in conjunction with the chosen $\mathbf{W}$, are omniscience-achieving. For this reason, we only prove the existence of a weight matrix $\mathbf{W}$ such that there are no unstable fixed modes in $\left( \mathbf{W} \otimes A, \overline{B}, \overline{H}\right)$.

	Without loss of generality, suppose $\mathcal{G}$ has 2 source components $\mathcal{G}_1 = \left( \mathbb{V}_1, \mathbb{E}_1 \right)$ and $\mathcal{G}_2=~\left( \mathbb{V}_2, \mathbb{E}_2 \right)$ and select a matrix $\mathbf{W}$ that satisfies \textbf{(F0)}-\textbf{(F4)} of Lemma \ref{lemma_weight_matrix}. In what follows, we verify the rank condition presented in Proposition \ref{prop_fixed_modes} to show that there are no unstable fixed modes in $\left( \mathbf{W} \otimes A, \overline{B}, \overline{H} \right)$. 

For any $\mathbb{J} \subseteq \mathbb{V}$ and its complement $\mathbb{J}^{c} = \mathbb{V} \setminus \mathbb{J}$, we define $\mathbb{J}_1 = \mathbb{V}_1 \cap \mathbb{J}$ and $\mathbb{J}_2 = \mathbb{V}_2 \cap \mathbb{J}$, and their complements $\mathbb{J}_1^{c} = \mathbb{V}_1 \setminus \mathbb{J}_1$ and $\mathbb{J}_2^{c} = \mathbb{V}_2 \setminus \mathbb{J}_2$, respectively. Also, for notational convenience, let $\mathbb{V}_3 = \mathbb{V} \setminus \left( \mathbb{V}_1 \cup \mathbb{V}_2 \right)$. Then, for $\lambda \in \Lambda_U(\mathbf{W} \otimes A)$, we can see that the following relation holds:
	
	  \allowdisplaybreaks\begin{align} \label{eq_thm_existence_02}
		 	&rank \left( \begin{array} {c | c} \mathbf{W} \otimes A - \lambda I_{|\mathbb{V}| \cdot n} & \overline{B}_{\mathbb{J}} \\ \hline
				\overline{H}_{\mathbb{J}^{c}} & \mathbf{0} \end{array} \right) \nonumber \\
			&\overset{\text{(i)}}{=} rank \left( \begin{array} {c | c} \begin{array} {c c c} \mathbf{W}_{1} \otimes A - \lambda I_{|\mathbb{V}_1| \cdot n} & \mathbf{0} & \mathbf{0} \\ \mathbf{0} & \mathbf{W}_{2} \otimes A - \lambda I_{|\mathbb{V}_2| \cdot n} & \mathbf{0}  \\ \mathbf{W}_{31} \otimes A & \mathbf{W}_{32} \otimes A & \mathbf{W}_{33} \otimes A - \lambda I_{|\mathbb{V}_3| \cdot n} \end{array} & \overline{B}_{\mathbb{J}} \\ \hline
				\overline{H}_{\mathbb{J}^{c}} & \mathbf{0} \end{array} \right) \nonumber \\
			&\overset{\text{(ii)}}{\geq} rank \left( \begin{array} {c | c} \begin{array} {c c c} \mathbf{W}_{1} \otimes A - \lambda I_{|\mathbb{V}_1| \cdot n} & \mathbf{0} & \mathbf{0} \\ \mathbf{0} & \mathbf{W}_{2} \otimes A - \lambda I_{|\mathbb{V}_2| \cdot n} & \mathbf{0}  \\ \mathbf{W}_{31} \otimes A & \mathbf{W}_{32} \otimes A & \mathbf{W}_{33} \otimes A - \lambda I_{|\mathbb{V}_3| \cdot n} \end{array} & \begin{array} {c c} \overline{B}_{\mathbb{J}_1} & \overline{B}_{\mathbb{J}_2} \end{array} \\ \hline
				\begin{array} {c} \overline{H}_{\mathbb{J}_1^{c}} \\ \overline{H}_{\mathbb{J}_2^{c}} \end{array} & \mathbf{0} \end{array} \right) \nonumber \\
		\begin{split}
			&\overset{\text{(iii)}}{=} rank \left( \begin{array} {c | c} \begin{array} {c c c} \mathbf{W}_{1} \otimes A - \lambda I_{|\mathbb{V}_1| \cdot n} & \mathbf{0} & \mathbf{0} \\ \mathbf{0} & \mathbf{0} & \mathbf{0}  \\ \mathbf{0} & \mathbf{0} & \mathbf{0} \end{array} & \overline{B}_{\mathbb{J}_1} \\ \hline
				\overline{H}_{\mathbb{J}_1^{c}} & \mathbf{0} \end{array} \right) 
			+ rank \left(\begin{array} {c | c} \begin{array} {c c c} \mathbf{0} & \mathbf{0} & \mathbf{0} \\ \mathbf{0} & \mathbf{W}_{2} \otimes A - \lambda I_{|\mathbb{V}_2| \cdot n} & \mathbf{0} \\ \mathbf{0} & \mathbf{0} & \mathbf{0} \end{array} & \overline{B}_{\mathbb{J}_2} \\ \hline 
				\overline{H}_{\mathbb{J}_2^{c}} & \mathbf{0} \end{array} \right) \\
				& \quad \quad + |\mathbb{V}_3| \cdot n
		\end{split}
	\end{align}
	
	In order to explain why the equalities and inequality in \eqref{eq_thm_existence_02} hold, we notice that  (i) follows directly from \textbf{(F0)} of Lemma \ref{lemma_weight_matrix}, (ii) holds by the fact that $\mathbb{J}_1, \mathbb{J}_2 \subseteq \mathbb{J}$ and $\mathbb{J}_1^{c}, \mathbb{J}_2^{c} \subseteq \mathbb{J}^{c}$, and (iii) holds by \textbf{(F4)} of Lemma \ref{lemma_weight_matrix} and by the definition of $\overline{B}$ and $\overline{H}$ in \eqref{eq_variable_def}.
	
	If $\mathbb{J}_1$ is not empty then by \textbf{(F1)}, \textbf{(F3)} of Lemma \ref{lemma_weight_matrix} and by the definition of $\overline{B}$ in \eqref{eq_variable_def}, it holds that 
	\begin{equation} \label{eq_thm_existence_03}
		rank \left( \begin{array} {c | c} \begin{array} {c c c} \mathbf{W}_{1} \otimes A - \lambda I_{|\mathbb{V}_1| \cdot n} & \mathbf{0} & \mathbf{0} \\ \mathbf{0} & \mathbf{0} & \mathbf{0}  \\ \mathbf{0} & \mathbf{0} & \mathbf{0} \end{array} & \overline{B}_{\mathbb{J}_1} \\ \hline
			\overline{H}_{\mathbb{J}_1^{c}} & \mathbf{0}  \end{array} \right) 
		\geq rank \left( \begin{array} {c | c} \begin{array} {c} \mathbf{W}_{1} \otimes A - \lambda I_{|\mathbb{V}_1| \cdot n} \\ \mathbf{0} \\ \mathbf{0} \end{array} & \overline{B}_{\mathbb{J}_1} \end{array} \right) = |\mathbb{V}_1| \cdot n.
	\end{equation}
	Otherwise, since $\mathbb{J}_1^{c} = \mathbb{V}_1$ and $\left(A, H_{\mathbb{V}_1} \right)$ is detectable (by the detectability condition of Theorem~\ref{thm_main}), by \textbf{(F1)}, \textbf{(F3)}, and the definition of $\overline{H}$ in \eqref{eq_variable_def}, it holds that
	\begin{equation} \label{eq_thm_existence_04}
		rank \left( \begin{array} {c | c} \begin{array} {c c c} \mathbf{W}_{1} \otimes A - \lambda I_{|\mathbb{V}_1| \cdot n} & \mathbf{0} & \mathbf{0} \\ \mathbf{0} & \mathbf{0} & \mathbf{0} \\ \mathbf{0} & \mathbf{0} & \mathbf{0} \end{array} & \overline{B}_{\mathbb{J}_1} \\ \hline
			\overline{H}_{\mathbb{J}_1^{c}} & \mathbf{0} \end{array} \right) 
		= rank \left( \begin{array} {c} \begin{array} {c c c} \mathbf{W}_{1} \otimes A - \lambda I_{|\mathbb{V}_1| \cdot n} & \mathbf{0} & \mathbf{0} \end{array} \\  \hline \overline{H}_{\mathbb{V}_1} \end{array} \right) = |\mathbb{V}_1| \cdot n.
	\end{equation}
	
	A similar relation holds for the second term in the last line of \eqref{eq_thm_existence_02}. Thus, by \eqref{eq_thm_existence_02}-\eqref{eq_thm_existence_04}, we conclude that for any subset $\mathbb{J} \subseteq \mathbb{V}$ and its complement $\mathbb{J}^{c} = \mathbb{V} \setminus \mathbb{J}$, it holds that
	\begin{equation}
	 	\begin{split}
		 	rank \begin{pmatrix} \mathbf{W} \otimes A - \lambda I_{|\mathbb{V}| \cdot n} & \overline{B}_{\mathbb{J}} \\ \overline{H}_{\mathbb{J}^{c}} & 0 \end{pmatrix} \geq |\mathbb{V}_1| \cdot n + |\mathbb{V}_2| \cdot n + |\mathbb{V}_3| \cdot n = |\mathbb{V}| \cdot n.
		\end{split}
	\end{equation}
	The non-existence of unstable fixed modes in $(\mathbf{W} \otimes A, \overline{B}, \overline{H} )$ follows from Proposition \ref{prop_fixed_modes}.
		
	\textbf{Necessity:} Without loss of generality, we suppose that a subsystem $\left( A, H_{\mathbb{V}_1} \right)$ associated with a source component $\mathcal{G}_1 = \left( \mathbb{V}_1, \mathbb{E}_1 \right)$ of $\mathcal{G}$ is not detectable. We will show that the observers represented by the vertices in $\mathbb{V}_1$ cannot achieve omniscience.

	Let $\mathbb{V}_1 = \{1, \cdots, m_1\}$, and $\mathbf{W}_1$ be any matrix whose sparsity pattern is consistent with $\mathcal{G}_1$ and that satisfies $\mathbf{W}_1 \cdot \mathbf{1} = \mathbf{1}$. Since $\mathcal{G}_1$ is a source component, there is no incoming edge to $\mathbb{V}_1$ from $\mathbb{V} \setminus \mathbb{V}_1$ in $\mathcal{G}$; hence, as can be seen in \eqref{eq_error_dynamics_col} the estimation error $\widetilde{x}_i$ and the augmented state $z_i$ for $i \in \mathbb{V}_1$ do not depend on $\widetilde{x}_j$ and $z_j$ for $j \in \mathbb{V} \setminus \mathbb{V}_1$. For this reason, the portion of the state-space representation of the error dynamics \eqref{eq_error_dynamics_col} associated with $\mathcal{G}_1$ can be written as follows:
	\begin{equation} \label{eq_thm_existence_05}
		\begin{split}
			\begin{pmatrix} \widetilde{x}' (k+1) \\ z' (k+1)\end{pmatrix} = \begin{pmatrix} \mathbf{W}_{1} \otimes A - \overline{B}' \overline{\mathbf{K}}' \overline{H}' & - \overline{B}' \overline{\mathbf{P}}' \\ \overline{\mathbf{Q}}' \overline{H}' & \overline{\mathbf{S}}' \end{pmatrix} \begin{pmatrix} \widetilde{x}' (k) \\ z' (k)\end{pmatrix}
		\end{split}
	\end{equation}
	with
	\begin{equation*}
		\begin{split}
			&\widetilde{x}' = \left( \widetilde{x}_1^T, \cdots, \widetilde{x}_{m_1}^T \right)^T, \quad z' = \left( z_1^T, \cdots, z_{m_1}^T \right)^T \\
			&\overline{B}' = \begin{pmatrix} \overline{B}_1' & \cdots & \overline{B}_{m_1}' \end{pmatrix} \text{ with } \overline{B}_i' = e_i' \otimes I_n \\
			&\overline{H}' = \begin{pmatrix} \left(\overline{H}_1'\right)^T & \cdots & \left(\overline{H}_{m_1}'\right)^T \end{pmatrix}^T \text{ with } \overline{H}_i' = \left(e_i'\right)^T \otimes H_i \\
			&\overline{\mathbf{K}}' = diag \left(\mathbf{K}_1, \cdots, \mathbf{K}_{m_1} \right),
			\quad \overline{\mathbf{P}}' = diag \left(\mathbf{P}_1, \cdots, \mathbf{P}_{m_1} \right) \\
			&\overline{\mathbf{Q}}' = diag \left(\mathbf{Q}_1, \cdots, \mathbf{Q}_{m_1} \right),
			\quad \overline{\mathbf{S}}' = diag \left(\mathbf{S}_1, \cdots, \mathbf{S}_{m_1} \right), \\
		\end{split}
	\end{equation*}
	where $e_i'$ is the $i$-th column of the $m_1$-dimensional identity matrix.

	Since we assume that the subsystem $\left( A, H_{\mathbb{V}_1} \right)$ is not detectable, it holds that \\ $rank \begin{pmatrix} \mathbf{W}_{1} \otimes A - \lambda I_{|\mathbb{V}_1| \cdot n} \\ \overline{H}_{\mathbb{V}_1}' \end{pmatrix} < |\mathbb{V}_1| \cdot n$ for some $\lambda \in \Lambda_U(\mathbf{W}_{1} \otimes A)$. Hence, by Proposition~\ref{prop_fixed_modes}, no gain parameters $\left\{ \mathbf{K}_{i}, \mathbf{P}_{i}, \mathbf{Q}_{i}, \mathbf{S}_{i} \right\}_{i \in \mathbb{V}_1}$ stabilize \eqref{eq_thm_existence_05}. Since $\mathbf{W}_1$ is chosen arbitrarily, we conclude that omniscience-achieving parameters do not exist, and the observers represented by the vertices in $\mathbb{V}_1$ cannot achieve omniscience. This proves the necessity. \hfill \QED
		
\section{Conclusions} \label{sec_conclusion}
We described a parametrized class of LTI distributed observers for state estimation of a LTI plant, where information exchange among observers is constrained by a pre-selected communication graph. We developed necessary and sufficient conditions for the existence of parameters for a LTI distributed observer that achieves asymptotic omniscience. These conditions can be described by the detectability of certain subsystems of the plant that are associated with source components of the communication graph.

\begin{appendices}
\section{The Proof of Lemma \ref{lemma_unique_rep}} \label{appen_lemma_proof}
	\textbf{Proof of (i)}: Let $\{\mu_1, \cdots, \mu_s\}$ and $\{\lambda_1, \cdots, \lambda_t\}$ be the sets of  distinct non-zero eigenvalues of $A$ and $L$, respectively. Under the choice $\mathbf{W}=I_m - \alpha L$, we can observe that if the UEPP of $\mathbf{W} \otimes A$ does not hold then $(1- \alpha \lambda) \mu = (1- \alpha \lambda') \mu'$ for some $\lambda, \lambda' \in \{\lambda_1, \cdots, \lambda_t\}$ and $\mu, \mu' \in \{\mu_1, \cdots, \mu_s\}$ where $\lambda \neq \lambda'$ and $\mu \neq \mu'$. Since the sets of distinct eigenvalues of $A$ and $L$ are both finite, we conclude that the set of values of $\alpha$ for which the UEPP does not hold is finite. Hence, for almost every positive number $\alpha$, $\mathbf{W} \otimes A$ satisfies the UEPP.
	
	\textbf{Proof of (ii)}: If $L$ is a WLM then for sufficiently small $\alpha > 0$, we can see that $\mathbf{W}=I_m - \alpha L$ becomes a stochastic matrix. Thus, using the proof of \textbf{(i)}, we conclude that, for sufficiently small $\alpha > 0$ except for finitely many points, $\mathbf{W}$ becomes a stochastic matrix and $\mathbf{W} \otimes A$ satisfies the UEPP. \hfill \QED

\section{The Proof of Theorem \ref{thm_wlm}} \label{appen_weighted_laplacian}
In this section, we provide a proof of Theorem \ref{thm_wlm}. The proof hinges on some results on structured linear system theory \cite{dion2003_automatica, reinschke1988_springer}. To this end, we briefly review the structural controllability and observability of structured linear systems in Appendix \ref{subsection_struct_cont_observ} and provide the detailed proof of Theorem \ref{thm_wlm} in Appendix \ref{subsection_proof_thm_wlm}.

\subsection{Structural Controllability and Observability} \label{subsection_struct_cont_observ}
Consider a graph $\overline{\mathcal{G}} = \left( \overline{\mathbb{V}}, \overline{\mathbb{E}} \right)$ with $\overline{\mathbb{V}}=\{1, \cdots, m\}$ and an associated structured linear system described as follows:
\begin{equation} \label{eq_struct_linear_system}
	\begin{split}
		x(k+1) & = [A] x(k) + [b_i] u(k) \\
		y(k) & = [h_j]^T x(k)
	\end{split}
\end{equation}
where $[A] \in \mathbb{R}^{m \times m}$ is a structure matrix, and $[b_i] \in \mathbb{R}^{m}$ and $[h_j] \in \mathbb{R}^{m}$ are structure vectors. Based on respective sparse structures, the elements of the structure matrix and vectors are either zero or indeterminate. In particular, $[A]$ is consistent with the graph $\overline{\mathcal{G}}$\footnote{The structure matrix $[A]$ is consistent with the graph $\overline{\mathcal{G}}$ if the following hold: (i) the $(i,j)$-th element of $[A]$ is zero if $(j,i) \notin \overline{\mathbb{E}}$ and (ii) the $(i,j)$-th element of $[A]$ is indeterminate if $(j,i) \in \overline{\mathbb{E}}$.}, and $[b_i]$ and $[h_j]$ are vectors whose elements are all zero except the $i$-th element and $j$-th element, respectively. Then, there are $|\overline{\mathbb{E}}|+2$ indeterminate elements of $[A]$, $[b_i]$, and $[h_j]$, and these indeterminate elements can be represented by vectors in $\mathbb{R}^{|\overline{\mathbb{E}}|+2}$. In other words, the vectors in $\mathbb{R}^{|\mathbb{E}|+2}$ specify all \textit{numerical realizations} of the structured linear system \eqref{eq_struct_linear_system}. 

The following Definition describes the structural controllability and observability of a structured linear system. As the underlying (sparse) structure of a structured linear system depends on its associated graph, we can characterize the structural controllability and observability in terms of the associated graph, which is specified in Proposition~\ref{prop_struct_cont_observ}.

\begin{mydef}
	Let  a graph $\overline{\mathcal{G}} = \left( \overline{\mathbb{V}},\overline{\mathbb{E}} \right)$ and an associated structured linear system as in \eqref{eq_struct_linear_system} be given. Let $p \in \mathbb{R}^{|\overline{E}|+2}$ be a vector that specifies a numerical realization of the structured system. The pair $\left( [A], [b_i] \right)$ is said to be \textit{structurally controllable} if for almost all $p \in \mathbb{R}^{|\overline{E}|+2}$, the resultant numerical realization of $\left( [A], [b_i] \right)$ is controllable. The \textit{structural observability} is similarly defined for the pair $\left( [A], [h_j]^T\right)$.
\end{mydef}

\begin{prop} \label{prop_struct_cont_observ}
	Let a graph $\overline{\mathcal{G}} = \left( \overline{\mathbb{V}},\overline{\mathbb{E}} \right)$ and an associated structured linear system as in \eqref{eq_struct_linear_system} be given. If $\overline{\mathcal{G}}$ is strongly connected and all its vertices have a loop, i.e, $(i, i) \in \overline{\mathbb{E}}$, then for all $i,j \in \overline{\mathbb{V}}$, the pair $\left( [A], [b_i] \right)$ is structurally controllable and the pair $\left( [A], [h_j]^T\right)$ is structurally observable.
\end{prop}
\begin{proof}
	The proof directly follows from relevant results from the structured linear system literature (see, for instance, Theorem 1 in \cite{dion2003_automatica}). The detail is omitted for brevity.
\end{proof}

\subsection{The Proof of Theorem \ref{thm_wlm}} \label{subsection_proof_thm_wlm}
	The following Lemma is used in the proof of Theorem \ref{thm_wlm}.
\begin{lemma} \label{lemma_wlm_graph}
	Given a strongly connected graph $\mathcal{G} = \left( \mathbb{V}, \mathbb{E} \right)$, the following hold for a fixed vertex $r \in \mathbb{V}$:
	\begin{itemize}
		\item[\textbf{(i)}] There exists $L_1 \in \mathbb{L}(\mathcal{G})$ for which the pair $\left( L_1, e_r^T \right)$ is observable.
		\item[\textbf{(ii)}] There exists $L_2 \in \mathbb{L}(\mathcal{G})$ for which the pair $\left( L_2, e_r \right)$ is controllable.
		\item[\textbf{(iii)}] There exists $L_3 \in \mathbb{L}(\mathcal{G})$ for which all the eigenvalues of $L_3$ are simple.
	\end{itemize}
where $e_r$ is the $r$-th column of the $m$-dimensional identity matrix.
\end{lemma}
\begin{proof}
	The proof is in two parts: In the first part, we prove \textbf{(i)} and \textbf{(ii)} using Proposition~\ref{prop_struct_cont_observ} (in Appendix~\ref{subsection_struct_cont_observ}), and then we provide a constructive proof of \textbf{(iii)}.

	\textbf{Proof of (i) and (ii):} 
	Let $\overline{\mathcal{G}} = \left( \overline{\mathbb{V}}, \overline{\mathbb{E}} \right)$ be a graph that extends $\mathcal{G}$ in the following way: $\overline{\mathbb{V}} = \mathbb{V}$ and $\overline{\mathbb{E}} =~\mathbb{E} \cup~\left( \bigcup_{i \in \mathbb{V}} (i,i)\right)$, i.e., $\overline{\mathcal{G}}$ is precisely same as $\mathcal{G}$ except every vertex of $\overline{\mathcal{G}}$ has a loop. Consider a structured linear system $\left( [A], [b_r], [h_r]^T \right)$ that is associated with $\overline{\mathcal{G}}$ as in \eqref{eq_struct_linear_system}. By Proposition \ref{prop_struct_cont_observ}, we can find numerical realizations $\left(A_1, h_r^T \right)$ and $\left(A_2, b_r \right)$ that are respectively observable and controllable. In particular, we may choose $A_1$ and $A_2$ to be (element-wise) nonnegative.

	We compute $L_1$ from $A_1$ by applying a special similarity transform used in \cite{pasqualetti2012_ieee_tac}. This procedure is described as follows: By the Perron-Frobenius Theorem, we can find a right eigenvector $\tilde{v}$ (of $A_1$) with all positive entries, which corresponds to the Perron-Frobenius eigenvalue $\tilde{\lambda}$. Let $M$ be a diagonal matrix whose diagonal elements are the entries of $\tilde{v}$. Then, by applying a similarity transform to $\left(A_1, h_r^T \right)$ with $M$, we obtain $\left( M^{-1} A_1 M, h_r^T M\right)$. Since the observability is preserved under any similarity transform, the new pair $\left( M^{-1} A_1 M, h_r^T M\right)$ is also observable. Note that $M^{-1} A_1 M$ and $A_1$ have the same sparsity pattern, and so do $h_r^T$ and $h_r^T M$. 

	Let's define
	\begin{equation}
		\begin{split}
			L_1 &= I - \frac{1}{\tilde{\lambda}} M^{-1} A_1 M \\
		\end{split}
	\end{equation}
Notice that $L_1$ belongs to $\mathbb{L}(\mathcal{G})$, and that the eigenvectors of $L_1$ are same as those of $M^{-1} A_1 M$. Since $\left( M^{-1} A_1 M, h_r^T M\right)$ is observable, by the PBH rank test, we can see that $\left( L_1, e_r^T\right)$ is observable, where $e_r$ is the $r$-th column of the $m$-dimensional identity matrix.

	By a similar argument, we can explicitly find $L_2 \in \mathbb{L}(\mathcal{G})$ for which $\left( L_2, e_{r}\right)$ is a controllable pair. This proves the first part of the proof.

	\textbf{Proof of (iii):} Given $L \in \mathbb{L}(\mathcal{G})$, we represent $L$ as follows:
\begin{equation}
	L = \begin{pmatrix} l_1 & \cdots & l_m \end{pmatrix}^T,
\end{equation}
where $l_i^T$ is the $i$-th row of $L$. By re-scaling each row of $L$, we construct $L_3 \in \mathbb{L}(\mathcal{G})$ whose eigenvalues are all simple. 

First of all, it is not difficult to show that the following matrix has all simple eigenvalues except at the origin for $\alpha_1 > 0$.
\begin{equation} \label{eq_lemma_wlm_graph_induction_1}
	\begin{pmatrix} \alpha_1 l_1 & \mathbf{0} & \cdots & \mathbf{0} \end{pmatrix}^T \in \mathbb{R}^{m \times m},
\end{equation}
where $\mathbf{0} \in \mathbb{R}^{m}$ is the $m$-dimensional zero vector.

Suppose that the following matrix has all simple eigenvalues except at the origin for some $\alpha_i > 0, ~ i \in \{1, \cdots, k\}$.
\begin{equation} \label{eq_lemma_wlm_graph_induction_2}
	\begin{pmatrix} \alpha_1 l_1 & \cdots & \alpha_k l_k & \mathbf{0} & \cdots & \mathbf{0} \end{pmatrix}^T \in \mathbb{R}^{m \times m}
\end{equation}

Recall that the eigenvalues of a matrix depend continuously on the elements of the matrix. Since \eqref{eq_lemma_wlm_graph_induction_2} has all simple eigenvalues except at the origin, for sufficiently small $\alpha_{k+1} > 0$, the following matrix has all simple eigenvalues except at the origin.
\begin{equation} \label{eq_lemma_wlm_graph_induction_3}
	\begin{pmatrix} \alpha_1 l_1 & \cdots & \alpha_{k} l_{k} & \alpha_{k+1} l_{k+1} & \mathbf{0} & \cdots &\mathbf{0} \end{pmatrix}^T \in \mathbb{R}^{m \times m}
\end{equation}

By mathematical induction, we obtain
\begin{equation} \label{eq_lemma_wlm_graph_induction_3}
	L_3 = \begin{pmatrix} \alpha_1 l_1 & \cdots & \alpha_m l_m \end{pmatrix}^T \in \mathbb{L} \left( \mathcal{G} \right)
\end{equation}
such that $L_3$ has all simple eigenvalues except at the origin, where $\alpha_i > 0, ~ i \in \{1, \cdots, m\}$. Since $\mathcal{G}$ is a strongly connected graph, the eigenvalue of $L_3$ at the origin is also simple \cite{caughman2006_ej_combina}. This completes the last part of the proof.
\end{proof}

\quad\textit{Proof of Theorem \ref{thm_wlm}:} To begin with, for the given graph $\mathcal{G} = \left( \mathbb{V}, \mathbb{E} \right)$, we define the following sets and a natural bijective mapping:
	\begin{equation*}
		\begin{split}
			&\mathbb{L}_{1,r}^{c}(\mathcal{G}) \overset{def}{=} \left\{L \in \mathbb{L}(\mathcal{G}) \mid \left( L, e_r^T\right) \text{ is not observable} \right\} \\
			&\mathbb{L}_{2,r}^{c}(\mathcal{G}) \overset{def}{=} \left\{L \in \mathbb{L}(\mathcal{G}) \mid \left( L, e_r\right) \text{ is not controllable} \right\} \\
			&\mathbb{L}_1^c (\mathcal{G}) \overset{def}{=} \left\{ L \in \mathbb{L}(\mathcal{G}) \mid \text{A right eigenvector of $L$ has a zero entry} \right\} \\
			&\mathbb{L}_2^c (\mathcal{G}) \overset{def}{=} \left\{ L \in \mathbb{L}(\mathcal{G}) \mid \text{A left eigenvector of $L$ has a zero entry} \right\} \\
			&\mathbb{L}_3^c (\mathcal{G}) \overset{def}{=} \left\{ L \in \mathbb{L}(\mathcal{G}) \mid \text{An  eigenvalue of $L$ is not simple} \right\} \\
			&\pi: \mathbb{L}(\mathcal{G}) \rightarrow \mathbb{R}_{< 0} ^{|\mathbb{E}|},
		\end{split}
	\end{equation*}
	where $e_r$ is the $r$-th column of the $m$-dimensional identity matrix, and $\mathbb{R}_{< 0} ^{|\mathbb{E}|}$ is the set of $|\mathbb{E}|$-dimensional vectors whose entries are all negative. To prove Theorem \ref{thm_wlm}, it is sufficient to prove that $\pi \left( \mathbb{L}_1^c(\mathcal{G}) \right)$, $\pi \left( \mathbb{L}_2^c(\mathcal{G}) \right)$, and $\pi \left( \mathbb{L}_3^c(\mathcal{G}) \right)$ all have the Lebesgue measure zero in $\mathbb{R}_{< 0} ^{|\mathbb{E}|}$.

In \cite{tchon1983_kybernetika}, the observability is shown to be a \textit{generic property} of structured linear systems. In words, unless every realization of a given structured linear system is not observable, almost every realization is observable. Hence, we have that unless $\mathbb{L}_{1,r}^{c}(\mathcal{G}) = \mathbb{L}(\mathcal{G})$, $\pi\left( \mathbb{L}_{1,r}^{c}(\mathcal{G}) \right)$ has the Lebesgue measure zero in $\mathbb{R}_{< 0}^{|\mathbb{E}|}$. By a similar argument for the controllability of structured linear systems, we conclude that unless $\mathbb{L}_{2,r}^{c}(\mathcal{G}) = \mathbb{L}(\mathcal{G})$, $\pi\left( \mathbb{L}_{2,r}^{c}(\mathcal{G}) \right)$ has the Lebesgue measure zero in $\mathbb{R}_{< 0}^{|\mathbb{E}|}$.

	Since $\mathcal{G}$ is a strongly connected graph, by Lemma \ref{lemma_wlm_graph} (in Appendix \ref{subsection_struct_cont_observ}), we can show that for any $r \in \mathbb{V}$, $\mathbb{L}_{1,r}^{c}(\mathcal{G})$ and $\mathbb{L}_{2,r}^{c}(\mathcal{G})$ are proper subsets of $\mathbb{L}(\mathcal{G})$; hence, $\pi\left(\mathbb{L}_{1,r}^{c}(\mathcal{G})\right)$ and $\pi\left(\mathbb{L}_{2,r}^{c}(\mathcal{G})\right)$ have the Lebesgue measure zero in $\mathbb{R}_{< 0}^{|\mathbb{E}|}$. Since $\mathbb{L}_1^c(\mathcal{G}) = \bigcup_{r \in \mathbb{V}} \mathbb{L}_{1,r}^{c}(\mathcal{G})$ and $\mathbb{L}_2^c(\mathcal{G}) =~\bigcup_{r \in \mathbb{V}} \mathbb{L}_{2,r}^{c}(\mathcal{G})$, we conclude that $\pi\left(\mathbb{L}_1^c (\mathcal{G})\right)$ and $\pi\left(\mathbb{L}_2^c (\mathcal{G}) \right)$ have the Lebesgue measure zero in $\mathbb{R}_{< 0}^{|\mathbb{E}|}$.
			
	Next, to prove that $\pi\left(\mathbb{L}_3^c(\mathcal{G})\right)$ has the Lebesgue measure zero, we adopt the following argument from algebra. For a matrix $L \in \mathbb{R}^{m \times m}$, the solutions to a polynomial equation $$\Delta(\lambda) \overset{def}{=} det(L-\lambda I) = a_m \lambda^m + \cdots + a_1 \lambda + a_0 = 0$$ are all distinct if the discriminant $D(\Delta) \overset{def}{=} a_m^{2m-1} \prod_{1 \leq i < j \leq m} (\lambda_i - \lambda_j)^2$ is nonzero, where $\lambda_i$ and $\lambda_j$ are solutions to $\Delta(\lambda) = 0$. By a classical result in algebra, this particular discriminant can be written as a polynomial function of the coefficients of $\Delta(\lambda)$ and those of its derivative $\Delta'(\lambda)$. Since the coefficients of $\Delta(\lambda)$ and $\Delta'(\lambda)$ are polynomial functions of the elements of $L$, the discriminant $D(\Delta)$ is a polynomial function of the elements of $L$.

	Also, notice that for a polynomial function $\widetilde{D}$ defined on $\mathbb{R}^{m}$, the solutions $\left\{q \in \mathbb{R}^{m} \mid \widetilde{D}\left( q \right) = 0\right\}$ to the polynomial equation $\widetilde{D}\left( q \right) = 0$ form either the entire space $\mathbb{R}^{m}$ or a hypersurface in $\mathbb{R}^{m}$, which has the Lebesgue measure zero \cite{fogaty1969_benjamin}.
	
	Therefore, by the above arguments, it holds either $\mathbb{L}_3^{c}(\mathcal{G}) = \mathbb{L}(\mathcal{G})$ or $\pi\left(\mathbb{L}_3^{c}(\mathcal{G})\right)$ has the Lebesgue measure zero in $\mathbb{R}_{< 0}^{|\mathbb{E}|}$. For the strongly connected graph $\mathcal{G}$, we have seen from Lemma~\ref{lemma_wlm_graph} (in Appendix \ref{subsection_struct_cont_observ}) that there exists $L_3 \in \mathbb{L}(\mathcal{G})$ whose eigenvalues are all simple. Therefore, $\mathbb{L}_3^{c}(\mathcal{G})$ is a proper subset of $\mathbb{L}(\mathcal{G})$ and $\pi\left(\mathbb{L}_3^{c}(\mathcal{G})\right)$ has the Lebesgue measure zero in $\mathbb{R}_{< 0}^{|\mathbb{E}|}$. \hfill \QED
	
\end{appendices}

\bibliographystyle{IEEEtran}
\bibliography{IEEEabrv,spark2012_journal}

\end{document}